\documentclass[letterpaper,12pt]{amsart}
\usepackage{amsmath,amsfonts,amsthm,mathrsfs,amssymb,graphicx,subfigure,bbm}
\usepackage[margin=1in]{geometry}

\usepackage{fancyhdr}

\usepackage{color}

\definecolor{myblue}{rgb}{0.3, 0.0, 0.85}
\definecolor{myviolet}{rgb}{0.5, 0.0, 0.5}

\usepackage{hyperref}

\hypersetup{
	colorlinks=true,
	linkcolor=myblue,
	citecolor=myviolet,
}

\DeclareMathOperator*{\slim}{s-lim}
\DeclareMathOperator*{\wlim}{w-lim}

\theoremstyle{plain}

\newtheorem{lem}{Lemma}[section]

\newtheorem{defn}{Definition}[section]

\newtheorem*{rem}{Remark}

\newtheorem*{unthm}{Theorem}
\newtheorem*{unlem}{Lemma}
\newtheorem*{unprop}{Proposition}
\newtheorem*{uncor}{Corollary}

\newtheorem{snthm}{Theorem}
\newtheorem{snprop}{Proposition}
\newtheorem{snconj}{Conjecture}

\title{Dynamics  of Noncommutative Solitons II: Spectral Theory, Dispersive Estimates and Stability}
\author{August J. Krueger}
\address{110 Frelinghuysen Rd., Rutgers Math.\ Dept., Piscataway, NJ 08854, USA}
\email{ akrueger@physics.rutgers.edu}
\author{Avy Soffer}
\address{110 Frelinghuysen Rd., Rutgers Math.\ Dept., Piscataway, NJ 08854, USA}
\email{soffer@math.rutgers.edu}

\begin{document}

\maketitle

\markright{Dynamics  of Noncommutative Solitons II: Spectral Theory, Dispersive Estimates and Stability}

\tableofcontents

\begin{abstract}
We consider the Schr\"odinger equation with a (matrix) Hamiltonian given by a second order difference operator with nonconstant growing coefficients, on the half one dimensional lattice. This operator appeared first naturally in the construction and dynamics of noncommutative solitons in the context of noncommutative field theory. We completely determine the spectrum of the Hamiltonian linearized around a ground state soliton and prove the optimal decay rate of $t^{-1}\log^{-2}t$ for the associated time decay estimate. We use a novel technique involving generating functions of orthogonal polynomials to achieve this estimate.
\end{abstract}

\section{Introduction and Background}

The notion of noncommutative soliton arises when one considers the nonlinear Klein-Gordon equation (NLKG) for a field which is dependent on, for example, two ``noncommutative coordinates'', $x,y$, whose coordinate functions satisfy canonical commutation relations (CCR) $[X,Y]=i\epsilon.$ By going to a representation of the above canonical commutation relation, one can reduce the dynamics of the problem to an equation for the coefficients of an expansion in the Hilbert space representation of the above CCR, see e.g. \cite{DJN 1}\cite{DJN 2}\cite{GMS}. By restricting to rotationally symmetric functions the nocommutative deformation of the Laplacian reduces to a second order finite difference operator, which is symmetric, and with variable coefficient growing lik the lattice coordinate, at infinity. Therefore, this operator is unbounded, and in fact has continuous spectrum $[0,\infty)$. These preliminary analytical results, as well as additional numerical results, were obtained by Chen, Fr\"ohlich, and Walcher \cite{CFW}. The dynamics and scattering of the (perturbed) soliton can then be inferred from the NLKG with such a discrete operator as the linear part. We will be interested  in studying the dynamics of discrete NLKG and discrete NLS equations with these hamiltonians.


We will be working with a discrete Schr\"odinger operator $L_0$ which can be considered either a discretization or a noncommutative deformation of the radial 2D negative laplacian, $-\Delta^{\mathrm{2D}}_\mathrm{r} = -r^{-1}\partial_r r \partial_r$. We will briefly review both perspectives.

In 1D one may find a discrete Laplacian via
\begin{align*}
& x \in \mathbb{R} \  \xrightarrow{\ \mathrm{discrete} \ } \  n \in \mathbb{Z}, \quad -\Delta^{\mathrm{1D}} = -\partial^{2}_x \  \xrightarrow{\ \mathrm{discrete} \ } \  -D_{+}D_{-},
\end{align*}
where $D_{+}v(n) = v(n) - v(x), D_{-}v(n) = v(n) - v(n-1)$ are respectively the forward and backward finite difference operators. It is important to implement this particular combination of these finite difference operators due in order to ensure that the resulting discrete Laplacian is symmetric. In 2D one may find a discrete Laplacian via
\begin{align*}
&r = (x^{2} + y^{2})^{1/2} =  2 \rho^{1/2},\quad \rho \in \mathbb{R}_+ \  \xrightarrow{\ \mathrm{discrete} \ } \  n \in \mathbb{Z}_+, \\
&-\Delta^{\mathrm{2D}}_{\mathrm{r}} = -r^{-1}\partial_{r}r\partial_{r} = -\partial_\rho \rho \partial_\rho \quad \xrightarrow{\ \mathrm{discrete} \ } \quad -D_{+}MD_{-} = L_{0},
\end{align*}
where $Mv(n) = nv(n)$. For any 1D  continuous coordinate $x$ one may discretize a pointwise multiplication straightforwardly via $v^p(x) \  \xrightarrow{\ \mathrm{discrete} \ } \  v^p(n)$, where $n$ is a discrete coordinate.

One may also follow the so-called noncommutative space perspective. Here one considers the formal ``Moyal star deformation'' of the algebra of functions on $\mathbb{R}^{2}$:
\begin{align*}
\Phi_{1}\cdot\Phi_{2}(x,y) &= \Phi_{1}(x,y)\Phi_{2}(x,y) \\
\xrightarrow{\ \epsilon > 0\ } \quad \Phi_{1} \star \Phi_{2}(x, y) &= \exp[i(\epsilon / 2)(\partial_{x_1}\partial_{y_2} - \partial_{y_1}\partial_{x_2})] \Phi_{1}(x_1,y_1)\Phi_{2}(x_2,y_2)\lfloor_{ (x_j,y_j) = (x,y) }.
\end{align*}
One calls the coordinates, $x,y$, noncommutative in this context because the coordinate functions $X(x,y) = x$, $Y(x,y) = y$ satisfy a nontrivial commutation relation $X\star Y - Y\star X \equiv [X,Y] = i\epsilon$. This prescription can be considered equivalent to the multiplication of functions of $q, p$ in quantum mechanics where operator ordering ambiguities are set by the normal ordering prescription for each product. For $\Phi$ a deformed function of $r = (x^2 + y^2)^{1/2}$ alone: $\Phi = \sum_{n=0}^\infty v(n)\Phi_n$ where $v(n) \in \mathbb{C}$ and the $\{\Phi_{n}\}_{n=0}^\infty$ are distinguished functions of $r$: the projectors onto the eigenfunctions of the noncommutative space variant of quantum simple harmonic oscillator system. One may find for $\Phi$ a function of $r$ alone:
\begin{align*}
-\Delta^{\mathrm{2D}} \Phi &= -\Delta^{\mathrm{2D}}_{\mathrm{r}}\Phi = -r^{-1}\partial_{r}r\partial_{r}\Phi \\
\xrightarrow{\ \epsilon > 0\ }\quad \frac{2}{\epsilon} L_0\Phi_n &= \frac{2}{\epsilon} \left\{ \begin{array}{cc}
		- (n+1)\Phi_{n+1} + (2n + 1)\Phi_{n} - n \Phi_{n-1} &,\quad n > 0 \\
		- \Phi_{1} + \Phi_{0} &,\quad n = 0 .
	\end{array} \right.
\end{align*}
which may be transferred to $\frac{2}{\epsilon} L_0v(n)$, an equivalent action on the $v(n)$, due to the symmetry of $L_0$. Since the $\Phi_n$ are noncommutative space representations of projection operators on a standard quantum mechanical Hilbert space, they diagonalize the Moyal star product: $\Phi_m\star\Phi_n = \delta_{m,n}\Phi_n$. This property is shared by all noncommutative space representations of projection operators. Thereby products of the $\Phi_n$ may be transferred to those of the expansion coefficients: $v(n)v(n) = v^2(n)$.

See  B. Durhuus, T. Jonsson, and R. Nest (2001) and T. Chen, J. Fr\"ohlich, and J. Walcher (2003) for reviews of the two approaches. In the following we will work on a lattice explicitly so $x \in \mathbb{Z}_+$ will be a discrete spatial coordinate.

The principle of replacing the usual space with a noncommutative space (or space-time) has found extensive use for model building in physics and in particular for allowing easier construction of localized solutions, see e.g. \cite{fuzzy physics}\cite{NC soliton survey} for surveys. An example of the usefulness of this approach is that it may provide a robust procedure for circumventing classical nonexistence theorems for solitons, e.g. that of Derrick \cite{Derrick}. The NLKG variant of the equation we study here first appeared in the context of string theory and associated effective actions in the presence of background D-brane configurations, see e.g. \cite{GMS}. We have decided to look in a completely different direction. The NLS variant and its solitons can in principle be materialized experimentally with optical devices, suitably etched, see e.g. \cite{Segev review}. Thus the dynamics of NLS with such solitons may offer new and potentially useful coherent states for optical devices. Furthermore, we believe the NLS solitons to have special properties, in particular asymptotic stability as opposed to the conjectured asymptotic metastability of the NLKG solitons conjectured in \cite{CFW}.


We will be following a procedure for the proof of asymptotic stability which has become standard within the study of nonlinear PDE \cite{Avy NLS}. Crucial aspects of the theory and associated results were established by Buslaev and Perelman \cite{important results 1}, Buslaev and Sulem \cite{important results 2}, and Gang and Sigal \cite{important results 3}. Important elements of these methods are the dispersive estimates. Various such estimates have been found in the context of 1D lattice systems, for example see the work of A.I. Komech, E.A. Kopylova, and M. Kunze \cite{important results 4} and of I. Egorova, E. Kopylova, G. Teschl \cite{1D lattice decay estimates}, as well as the continuum 2D problem to which our system bears many resemblances, see e.g. the work of E. A. Kopylova and A.I. Komech \cite{2D}. Extensive results have been found on the asymptotic stability on solitons of 1D nonlinear lattice Schr\"odinger equations by F. Palmero et al. \cite{important results 5} and P.G. Kevrekidis, D.E. Pelinovsky, and A. Stefanov \cite{important results 6}. Important aspects of the application of these models to optical nonlinear waveguide arrays has been established by H.S. Eisenberg et al. \cite{important results 7}.



In \cite{paper 01} we focus on a key estimate that is needed for scattering and stability, namely the decay in time of the solution, at the optimal rate. Fortunately, in the generic case, we find it is integrable, given by $t^{-1}\log^{-2}t$. The proof of this result is rather direct, and employs the generating functions of the corresponding generalized eigenfunctions, to explicitly represent and estimate the resolvent of the hamiltonian at all energies. We also conclude the absence of positive eigenvalues and singular continuous spectrum.

Preliminary results for the scattering theory of the associated noncommutative waves and solitons were found by Durhuus and Gayral \cite{noncommutative scattering}. In particular they find local decay estimates for the associated noncommutative NLS. We utilize alternative methods and find local decay for both the free Schr\"odinger operator as well as a class of rank one perturbations thereof. An important element of this analysis is the study of the spectral properties of the free and perturbed Schr\"odinger operator. We extend the analysis of Chen, Fr\"ohlich, and Walcher \cite{CFW} and reproduce some of their results with alternative techniques.

In \cite{paper 02} we address the construction and properties of a family of ground state solitons. These stationary states satisfy a nonlinear eigenvalue equation, are positive, monotonically decaying and sharply peaked for large spectral parameter. The proof of this result follows directly from our spectral results in this paper by iteration for small data and root finding for large data. The existence and many properties of solutions for a similar nonlinear eigenvalue equation were found by Durhuus, Jonssen, and Nest \cite{DJN 1}\cite{DJN 2}. We utilize a simple power law nonlinearity for which their existence proofs do not apply. We additionally find estimates for the peak height, spatial decay rate, norm bounds, and parameter dependence.

In this paper we focus on deriving a decay rate estimate for the Hamiltonian which results from linearizing the original NLS around the soliton constructed in \cite{paper 02}. We determine the full spectrum of this operator, which is the union of a multiplicity 2 null eigenvalue and a real absolutely continuous spectrum. This establishes a well-defined set of modulation equations \cite{Avy NLS} and points toward the asymptotic stability of the soliton.

In the conclusion of this paper we describe how the results can be applied to prove stability of the soliton we constructed in \cite{paper 02}. The issue of asymptotic stability of NLS solitons has been sufficiently well-studied in such a broad context that the proof thereof is often considered as following straightforwardly from the appropriate spectral and decay estimates, of the kind found in this paper. We sketch how the theory of modulation equations established by Soffer and Weinstein \cite{Avy NLS} can be used to prove asymptotic stability. Chen, Fr\"ohlich, and Walcher conjectured that in the NLKG case the corresponding solitons are unstable but with exponentially long decay time: the so-called metastability property, see e.g. \cite{CFW}. There is a great deal of evidence to suggest that this is in fact the case but a proof has yet to be provided. This will be the subject of future work.

\section{Notation}

Let $\mathbb{Z}_+$ and $\mathbb{R}_+$ respectively be the nonnegative integers and nonnegative reals and $\mathscr{H} = \ell^2(\mathbb{Z}_+,\mathbb{C})$ the Hilbert space of square integrable complex functions, e.g. $v: \mathbb{Z}_+ \ni x \mapsto v(x) \in \mathbb{C}$, on the 1D half-lattice with inner product $( \cdot , \cdot )$, which is conjugate-linear in the first argument and linear in the second argument, and the associated norm $||\cdot||$, where $||v|| = (v,v)^{1/2}$, $\forall v\in\mathscr{H}$. Where the distinction is clear from context $||\cdot|| \equiv ||\cdot||_{\mathrm{op}}$ will also represent the norm for operators on $\mathscr{H}$ given by $||A||_{\mathrm{op}} = \sup_{v \in \mathscr{H}}||v||^{-1}||Av||$, for all bounded $A$ on $\mathscr{H}$. Denote the lattice $\ell^1$ norm by $||\cdot||_1$ where $||v||_1 = \sum_{x=0}^\infty|v(x)|$, $\forall v \in \ell^1(\mathbb{Z}_+,\mathbb{C})$.

We denote by $\otimes$ the tensor product and by $z \mapsto \overline{z}$ complex conjugation for all $z \in \mathbb{C}$. We write $\mathscr{H}^*$ for the space of linear functionals on $\mathscr{H}$: the dual space of $\mathscr{H}$. For every $v \in \mathscr{H}$ one has that $v^* \in \mathscr{H}^*$ is its dual satisfying $v^{*}(w) = (v,w)$ for all $v,w \in \mathscr{H}$.  For every operator $A$ on $\mathscr{H}$ we take $\mathcal{D}(A)$ as standing for the domain of $A$. For each operator $A$ on $\mathscr{H}$ define $A^*$ on $\mathscr{H}^*$ to be its dual and $A^\dag$ on $\mathscr{H}$ its adjoint such that $v^*(Aw) = A^*v^*(w) = (A^\dag v, w)$ for all $v \in \mathcal{D}(A^\dag)$ and all $w \in \mathcal{D}(A)$. Let $\{\chi_{x}\}_{x=0}^{\infty}$ be the orthonormal set of vectors such that $\chi_{x}(x) = 1$ and $\chi_{x_{1}}(x_{2}) = 0$ for all $x_{2} \ne x_{1}$. We write $P_{x} = \chi_{x} \otimes \chi^{*}_{x}$ for the orthogonal projection onto the space spanned by $\chi_{x}$.

We define $\mathscr{T}$ to be the topological vector space of all complex sequences on $\mathbb{Z}_+$ endowed with topology of pointwise convergence, $\mathcal{B}(\mathscr{H})$ to be the space of bounded linear operators on $\mathscr{H}$, and $\mathcal{L}(\mathscr{T})$ to be the space of linear operators on $\mathscr{T}$, endowed with the pointwise topology induced by that of $\mathscr{T}$. When an operator $A$ on $\mathscr{H}$ can be given by an explicit formula through $A(x_{1},x_{2}) = (\chi_{x_{1}},A\chi_{x_{2}}) < \infty$ for all $x_{1},x_{2} \in \mathbb{Z}_{+}$ one may make the natural inclusion of $A$ into $\mathcal{L}(\mathscr{T})$, the image of which will also be denoted by $A$. We consider $\mathscr{T}$ to be endowed with pointwise multiplication, i.e. the product $uv$ is specified by $(uv)(x)=u(x)v(x)$ for all $u,v \in \mathscr{T}$.

We represent the \emph{spectrum} of each $A$ on $\mathscr{H}$ by $\sigma(A)$. We term each element $\lambda \in \sigma(A)$ a \emph{spectral value}. We write $\sigma_{\mathrm{d}}(A)$ for the \emph{discrete spectrum}, $\sigma_{\mathrm{e}}(A)$ for the \emph{essential spectrum}, $\sigma_{\mathrm{p}}(A)$ for the \emph{point spectrum}, $\sigma_{\mathrm{ac}}(A)$ for the \emph{absolutely continuous spectrum}, and $\sigma_{\mathrm{sc}}(A)$ for the \emph{singularly continuous spectrum}. Should an operator $A$ satisfy the spectral theorem there exist scalar measures $\{\mu_{n}\}_{n=1}^{N}$ on $\sigma(A)$ which furnish the associated spectral representation of $\mathscr{H}$ for $A$ such that the action of $A$ is given by multiplication by $\lambda \in \sigma(A)$ on $\oplus_{n=1}^{N}L^{2}(\sigma(A),\mathrm{d}\mu_{k})$. If $\mathscr{H} = \oplus_{k=1}^{n}L^{2}(\sigma(A),\mathrm{d}\mu_{k})$ we term $n$ the \emph{generalized multiplicity} of $A$. For an operator of arbitrary generalized multiplicity we will write $\mu^{A}$ for the associated operator valued measure, such that $A = \int_{\sigma(A)} \lambda\ \mathrm{d}\mu^{A}_{\lambda}$. For each operator $A$ that satisfies the spectral theorem, its spectral (Riesz) projections will be written as $P^{A}_{\mathrm{d}}$ and the like for each of the distinguished subsets of the spectral decomposition of A. Define $R^A_\cdot: \rho(A) \to \mathcal{B}(\mathscr{H})$, the resolvent of $A$, to be specified by $R^A_z := (A - z)^{-1}$, where $\rho(A) := \mathbb{C} \setminus \sigma(A)$ is the resolvent set of $A$ and where by abuse of notation $zI \equiv z \in \mathcal{B}(\mathscr{H})$ here.

Allow an \emph{eigenvector} of $A$ to be a vector $v \in \mathscr{H}$ for which $Av = \lambda v$ for some $\lambda \in \mathbb{C}$. Should $A$ admit inclusion into $\mathcal{L}(\mathscr{T})$, we define a \emph{generalized eigenvector} of $A$ be a vector $\phi \in \mathscr{T} \setminus \mathscr{H}$ which satisfies $A\phi = \lambda \phi$ for some $\lambda \in \mathbb{C}$ such that $\phi(x)$ is polynomially bounded, which is to say that there exists a $p \ge 0$ such that $\lim_{x \nearrow \infty}(x+1)^{-p}\phi(x) = 0$. We define a \emph{spectral vector} of $A$ to be a vector which is either an eigenvector or generalized eigenvector of $A$. We define the subspace of spectral vectors associated to the set $\Sigma \subseteq \sigma(A)$ to be the \emph{spectral space over $\Sigma$}.

We write $\partial_z \equiv \frac{\partial}{\partial z}$ and $\mathrm{d}_z \equiv \frac{\mathrm{d}}{\mathrm{d} z}$ respectively for formal partial and total derivative operators with respect to a parameter $z \in \mathbb{R}, \mathbb{C}$.

\section{Review}

\begin{defn}
Define $L_0$ to be the operator on $\mathscr{H}$ with action
\begin{align}
	L_0v(x) = \left\{
	\begin{array}{cc}
		- (x+1)v(x+1) + (2x + 1)v(x) - x v(x-1) &,\quad x > 0 \\
		- v(1) + v(0) &,\quad x = 0 .
	\end{array} \right.
\end{align}
and domain $\mathcal{D}(L_0) := \{ v \in \mathscr{H}\ |\ || Mv || < \infty \}$, where $M$ is the multiplication operator with action $Mv(x) = xv(x)$ $\forall v \in \mathscr{T}$.
\end{defn}



In \cite{paper 01} we proved the following.

\begin{snprop}\label{snprop01}
The operator $L_0$ has the following properties.
\begin{enumerate}
	\item $L_0$ is essentially self-adjoint.
	\item $L_{0}$ has generalized multiplicity 1.
	\item The spectrum of $L_0$ is absolutely continuous, $\sigma(L_0) = \sigma_{\mathrm{ac}}(L_0) = [0,\infty)$, and for choice of normalization $\phi_{\lambda}(0) = 1$, its generalized eigenfunctions are the Laguerre polynomials $\phi_{\lambda} = \sum_{k=0}^x \frac{(-\lambda)^k}{k!}\binom{x}{k}$.
\end{enumerate}
\end{snprop}

\noindent Chen, Fr\"ohlich, and Walcher determined the above properties for $L_{0}$ in \cite{CFW} via methods which are different from ours.

\begin{defn}
Let $w_{\lambda} := (\chi_{0},\delta^{L_{0}}_{\lambda} \chi_{0})$ be termed the \emph{spectral integral weight}, $\psi_z := R^{L_{0}}_z\chi_{0}$ the \emph{resolvent vector}, $\xi_z := \psi_z - \psi_z(0)\phi_z$ the \emph{auxilliary resolvent vector}, and $f_{z} = (\chi_{0},R^{L_{0}}_{z}\chi_{0})$ the \emph{resolvent function of $L_{0}$} for all $\lambda \in \sigma(L_{0})$ for all $z \in \rho(L_{0})$.
\end{defn}

Since $\phi_{\lambda}(x)$ is a polynomial of degree $x$ in $\lambda$, one has that the analytic continuation $\phi_{z} \in \mathscr{T}$, $z \in \mathbb{C}$, exists. The above permits the useful representation $\psi_{z} = f_{z}\phi_{z}+\xi_{z}$.

\begin{defn}
Define $L$ to be the operator on $\mathscr{H}$ with domain $\mathcal{D}(L) = \mathcal{D}(L_0)$ and specified by $L := L_0 - q P_0$ where $q \ge 0$ is a fixed constant. Let $\psi^{L}_{z} := R^{L}_{z}\chi_{0}$ be the resolvent vector of $L$ for all $z \in \rho(L)$.
\end{defn}

\begin{snthm}\label{snthm01}
Let $\phi^L_\lambda$, $\lambda \in \sigma(L)$, denote spectral vectors of $L$ chosen to satisfy the normalization condition $(\chi_{0},\phi^{L}_{\lambda}) = \phi^L_\lambda(0) = 1$, $\forall \lambda \in \sigma(L)$. $L$ has the following properties.
	\begin{enumerate}
		\item $\sigma_{\mathrm{d}}(L) = \sigma_{\mathrm{p}}(L) = \{ \lambda_0 \}$, where $\lambda_0 < 0$ uniquely satisfies $1 = q \psi_{\lambda_0}(0)$ and the unique eigenfunction over $\lambda_0$ is $\psi^L_{\lambda_0} = q\psi_{\lambda_0}$.
		\item $\sigma_{\mathrm{e}}(L) = \sigma_{\mathrm{ac}}(L) = \sigma(L_{0}) = [0,\infty)$ and has generalized multiplicity 1.
		\item $\mathrm{d}\mu^L(\lambda) = w^L_{\lambda}\phi^L_{\lambda} \otimes \phi^{L,*}_{\lambda} \ \mathrm{d}\lambda$, where $w^L_{\lambda} = \{ [1+ q e^{-\lambda} \mathrm{Ei}(\lambda)]^2 + [\pi q e^{-\lambda}]^2 \}^{-1} e^{-\lambda}$, $\mathrm{d}\lambda$ is the Lebesgue measure on $[0,\infty)$, and $\mathrm{Ei}(\lambda) := \int_{-\lambda}^\infty \!\mathrm{d}u\ u^{-1}e^{-u}$, $\lambda>0$, is the exponential integral. The generalized eigenfunctions of $L$ are given by $\phi^L_\lambda = \phi_{\lambda} + q \xi_{\lambda}$, $\lambda \in \sigma_{\mathrm{ac}}(L)$.
	\end{enumerate}
\end{snthm}

\begin{defn}
Let $W_{\kappa,\tau}$ be the multiplication operator weight specified by \\ $W_{\kappa,\tau}v(x) = (x + \kappa)^{\tau}v(x)$, $\forall v \in \mathscr{T}$, where $0 < \kappa \in \mathbb{R}$, $ \tau \in \mathbb{R}$.
\end{defn}

\begin{uncor}[1 of \cite{paper 01}]\label{uncor01}
One has that
\begin{align}
	\left|\mathrm{d}^{n}_{\lambda}\left[w^{1/2}_{\lambda}\phi_{\lambda}(x)\right]\right| < c(\phi_{\lambda},n)\exp(-16^{-1}\epsilon\lambda),\quad n = 0, 1, 2,
\end{align}
where $c(\phi_{\lambda},n) := 3^{n+1}(x+\kappa)^{n+1}$ and
\begin{align}
	\left|\mathrm{d}^{n}_{\lambda}\left[w^{1/2}_{\lambda}\xi_{\lambda}(x)\right]\right| < c(\xi_{\lambda},n)\exp(-16^{-1}\epsilon\lambda),\quad n = 0, 1, 2,
\end{align}
where
\begin{align}
	c(\xi_{\lambda},0) := 12(x+\kappa)^{2},\quad c(\xi_{\lambda},1) := 24(x+\kappa)^{2},\quad c(\xi_{\lambda},2) := 36(x+\kappa)^{3} ,
\end{align}
and $\epsilon = 1 - [1 - (x + \kappa)^{-1}]^{2}$.
\end{uncor}




In \cite{paper 02} we proved the following. Consider the discrete NLS
\begin{align}\label{NLS}
	i\partial_tw = L_0w - |w|^{2\sigma}w,\quad 1 \le \sigma \in \mathbb{Z}
\end{align}
where $w : \mathbb{R}_t \times \mathbb{Z}_+ \to \mathbb{C}$. The existence of a $u: \mathbb{Z}_+ \to \mathbb{C}$ which satisfies the nonlinear finite difference equation
\begin{align}
	L_0u = \zeta u + |u|^{2\sigma}u,
\end{align}
furnishes a stationary state of the discrete NLS of the form $w(t) = e^{-i\zeta t}u$. One expects that, due to the attractive nature of the nonlinearity, a negative ``nonlinear eigenvalue'', $\zeta = -a < 0$, will allow the existence of a sharply peaked, monotonically decaying ``ground state soliton''. We will therefore exclusively look for solutions to
\begin{align}\label{solitoneq}
	L_0u = -a u + u^{2\sigma + 1},
\end{align}
where $u: \mathbb{Z}_+ \to \mathbb{R}_{+}$ and $a > 0$. Solutions with these characteristics are self-focusing and tend to be sharply localized. They are therefore termed solitary waves or \emph{solitons} generally.

\begin{unthm}[1 of \cite{paper 02}]\label{unthm04}
There exists a $\mu_* > 0$ such that for each $ \mu > \mu_*$ there exists a solution to Equation \eqref{solitoneq} with $\zeta = -\mu < 0$ and $u = \alpha_{\mu}$, which is:
\begin{enumerate}
	\item positive: $\alpha_{\mu}(x) > 0$ for all $x \in \mathbb{Z}_+$
	\item monotonically decaying: $\alpha_{\mu}(x+1) - \alpha_{\mu}(x) < 0$ for all $x \in \mathbb{Z}_+$
	\item absolutely integrable: $\alpha_{\mu} \in \ell^1$
\end{enumerate}
\end{unthm}

\begin{unprop}[1 of \cite{paper 02}]\label{unprop02}
\begin{align}
	||(I-P_{0})\alpha_{\mu}||_{1} \le \mu^{-(2\sigma)^{-1}(2\sigma-1)} + \mathcal{O}(\mu^{-(2\sigma)^{-1}(4\sigma-1)}) .
\end{align}
\end{unprop}

\section{Results}

We would like to study the evolution of solutions which, at least at an initial time $t = 0$, are close to the stationary soliton solution $\widehat{u}(t) = e^{-i(-\mu t + \nu)}\alpha_{\mu}$, where $\nu \in \mathbb{R}$ is an arbitrary phase factor. We then consider the ansatz $u = e^{-i\theta}(\alpha_{\widehat{\mu}} + \beta)$ where
\begin{align}
	\theta(t) := - \int_{0}^{t} \widehat{\mu}(s)\mathrm{d}s + \widehat{\nu}(t),
\end{align}
\begin{align}
	\widehat{\mu}(t) = \mu + \widehat{\mu}_1(t),\quad \widehat{\nu}(t) = \nu + \widehat{\nu}_1(t),\quad \widehat{\mu}_{1}(0) = 0,\quad \widehat{\nu}_{1}(0) = 0 ,
\end{align}
and $\beta : \mathbb{R}_t \times \mathbb{Z}_+ \to \mathbb{C}$ has the property that it and $\partial_t\beta$ are small in norm at $t = 0$. If one finds that $u(t) \to e^{-i(-\mu_{\infty} t + \nu_{\infty})}\alpha_{\mu_{\infty}}$ in norm as $t \nearrow \infty$ for some $\mu_{\infty}$, $\nu_{\infty}$, for all $\mu$, $\nu$ and all sufficiently small $\beta$ then one calls $\widehat{u}(t)$ \emph{asymptotically stable}. The most important element of the proof of this analysis is the study of the spectral measure of the operator one obtains by linearizing around $\alpha_{\mu}$. One then considers the associated linearized NLS.

If $u = e^{-i\theta}(\alpha_{\widehat{\mu}} + \beta)$ satisfies the NLS then $\beta$ satisfies the linearized NLS (LNLS)
\begin{align}\label{LNLS}
	i\partial_t \vec{\beta} = H \vec{\beta} + \vec{\gamma},
\end{align}
where
\begin{align}
	H :&= (L_{0} + \mu)D - (\sigma+1)\alpha^{2\sigma}D - \sigma\alpha^{2\sigma}J , \\
	D :&= \begin{bmatrix} 1 & 0 \\ 0 & -1 \end{bmatrix}, \quad J := \begin{bmatrix} 0 & 1 \\ -1 & 0 \end{bmatrix},\quad \vec{\beta} := \begin{bmatrix} \beta \\ \overline{\beta} \end{bmatrix}, \quad \vec{\gamma} := \begin{bmatrix} \gamma \\ -\overline{\gamma} \end{bmatrix},
\end{align}
\begin{align}
	\alpha_{\widehat{\mu}} &= \alpha_\mu + \widehat{\alpha},\quad \alpha_{\mu} \equiv \alpha ,\quad \gamma = \gamma_{0} + \gamma_{1},\quad \gamma_{0} := -\mathrm{d}_t\widehat{\nu}\alpha_{\widehat{\mu}} + i\partial_{\widehat{\mu}}\alpha_{\widehat{\mu}}\mathrm{d}_{t}\widehat{\mu} , \\
	\gamma_{1} &:= \widehat{\mu}_1\beta -\mathrm{d}_{t}\widehat{\nu}\beta - [(\sigma+1)\beta +\sigma\overline{\beta}]\sum_{j=1}^{2\sigma}\binom{2\sigma}{j}\widehat{\alpha}^{j}\alpha^{2\sigma-j} \\
		&\quad\quad + \sum_{(j,k)}' \binom{\sigma}{j}\overline{\beta}^{j}\alpha^{\sigma-j}_{\widehat{\mu}}\binom{\sigma+1}{k}\beta^{k}\alpha^{\sigma+1-k}_{\widehat{\mu}},
\end{align}
where $\sum_{(j,k)}' $ sums over all $(j,k) \in \mathbb{Z}\times\mathbb{Z}$ for $0 \le j \le \sigma$ and $0 \le k \le \sigma+1$ with the exclusion of $(0,0)$, $(0,1)$ and $(1,0)$.

We will arrive at the properties of $H$ by studying a sequence of simpler operators.

\begin{defn}
\begin{align}
	& H_{0} := (L_{0} + \mu)D, \quad H_{1} := H_{0} - q_{1}P_{0} D, \quad H_{2} := H_{1} - q_{2}P_{0} J, \\
	&  \rho := \alpha(0), \quad q_{1} := (\sigma+1)\rho^{2\sigma}, \quad q_{2} := \sigma \rho^{2\sigma}, \\
	& L := L_{0} - q_{1}P_{0}, \quad U := H - H_{2}.
\end{align}
\end{defn}

\noindent $H, H_0, H_1, H_2, U$ act on $\vec{\mathscr{H}} := \mathscr{H} \oplus \mathscr{H}$, the natural extension of $\mathscr{H}$ to the matrix system. Although $H_{0}$ and $H_{1}$ are self-adjoint, it is the case that $H_{2}$ and $H$ are not. This property of $H$ is typical of linearized operators and makes the analysis very difficult for most systems.

\begin{snthm}\label{snthm05}
The spectrum of $H_{2}$ has the following properties.
\begin{enumerate}
	\item $\sigma_{\mathrm{d}}(H_{2}) = \sigma_{\mathrm{p}}(H_{2}) = \{ (-1)^{j} i (2\sigma)^{1/2}\mu^{-\sigma}[1+\mathcal{O}(\mu^{-1})] \}_{j=0}^{1}$.
	\item $\sigma_{\mathrm{e}}(H_{2}) = \sigma_{\mathrm{c}}(H_{2}) = \sigma_{\mathrm{ac}}(H_{2}) = (-\infty, -\mu] \cup [\mu, \infty)$.
\end{enumerate}
\end{snthm}

\begin{snthm}\label{snthm06}
The spectrum of $H$ has the following properties.
\begin{enumerate}
	\item $\sigma_{\mathrm{d}}(H) = \sigma_{\mathrm{p}}(H) = \{ 0 \}$, with multiplicity 2.
	\item $\sigma_{\mathrm{e}}(H) = \sigma_{\mathrm{c}}(H) = \sigma_{\mathrm{ac}}(H) = (-\infty, -\mu] \cup [\mu, \infty)$.
\end{enumerate}
\end{snthm}

\begin{snthm}\label{snthm07}
For all $-3 \ge \tau \in \mathbb{R}$, $v \in \ell^{1}$, there exists constants $0 < c$ and $1 < \kappa \in \mathbb{R}$ such that
\begin{align}
	|| W_{\kappa,\tau}e^{-itH_{2}}P^{H_{2}}_{\mathrm{e}}W_{\kappa,\tau}v ||_{\infty} = c t^{-1}\log^{-2}t ||v||_1,\quad t\nearrow\infty
\end{align}
\end{snthm}

\begin{snthm}\label{snthm08}
For all $-3 \ge \tau \in \mathbb{R}$, $W^{-1}_{\kappa,\tau}v \in \ell^{1}$, there exist constants $0 < c$ and $1 < \kappa \in \mathbb{R}$ such that
\begin{align}
	|| W_{\kappa,\tau}e^{-itH}P^{H}_{\mathrm{e}}W_{\kappa,\tau}v ||_{\infty} = c t^{-1}\log^{-2}t ||v||_1,\quad t\nearrow\infty .
\end{align}
\end{snthm}
It was shown by Soffer and Weinstein in \cite{Avy NLS II} that to prove decay in time of the form $t^{-1}\log^{-2}t$ for the linearized Hamiltonian with polynomial weights and no $\ell^p$ estimates is sufficient to prove asymptotic stability of the soliton.

One may extract from $E_1(z) := \int_1^\infty \mathrm{d}t\ e^{-z t} t^{-1}$ the well-known asymptotic expansion
\begin{align}
	E_{1}(z) = e^{-z}z^{-1}\sum_{k=1}^{n-1}{ k! \over (-z)^{k} } + \mathcal{O}(n!z^{-n}) ,
\end{align}
which is valid for large values of $\Re z$. This in turn gives
\begin{align}
	\psi_{z}(0) = z^{-1}\sum_{k=1}^{n-1}{ k! \over (-z)^{k} } + \mathcal{O}(n!z^{-n}),
\end{align}
which will be a crucial tool in our analysis.

\section{Spectral Properties of $H_{2}$}

Consider that $H = AD$ where $A$ is an essentially self-adjoint operator on $\mathscr{H}$ and $D$ is the diagonal matrix defined above, it is the case that $R^{H}_{z} = \begin{bmatrix} R^{A}_{z} & 0 \\ 0 & -R^{A}_{-z} \end{bmatrix}$ since
\begin{align}
	H - z &= \begin{bmatrix} A - z & 0 \\ 0 & - A - z \end{bmatrix} = \begin{bmatrix} A - z & 0 \\ 0 & - (A + z) \end{bmatrix}
\end{align}
so that for all $z \in \mathbb{C}$ such that $z,-z \in \rho(A)$ one may find by direct inversion
\begin{align}
	(H - z)^{-1} &= \begin{bmatrix} (A - z)^{-1} & 0 \\ 0 & - (A + z)^{-1} \end{bmatrix} .
\end{align}
One may write $R^{H_{1}}_{z} = \begin{bmatrix} R^{L}_{z_{1}} & 0 \\ 0 & -R^{L}_{z_{2}} \end{bmatrix}$, where here $L = L_{0} - q_{1}P_{0}$ and where $q_{1},z_{1}, z_{2}$ are defined as given above. Since $H^{2}_{2}$ and $H^{2}$ are self-adjoint, it follows that $\sigma(H_{2}),\sigma(H) \subseteq \mathbb{R} \cup i\mathbb{R}$, see e.g. \cite{Schlag}.

The first part of our analysis will be dedicated to proving that the point spectrum of $H_{2}$ consists of a conjugate pair of complex eigenvalues which are very close to the origin for large $\mu$.

\begin{lem}
The eigenvalues of $H_{2}$ are given by the roots of $h(z) := q_{2}^{2}f^{L}_{z_{1}}f^{L}_{z_{2}} - 1$, where $z_{1} := z - \mu$ and $z_{2} := - z - \mu$.
\end{lem}

\begin{proof}
Let $\vec{v} = \begin{bmatrix} v_{1} \\ v_{2} \end{bmatrix}$ and take $v_{1}(0) \neq 0$. One has that
\begin{align}
	H_{2}\vec{v} &= z\vec{v} \\
	(H_{1} - q_{2}P_{0}J)\vec{v} &= ,
\end{align}
then by inversion:
\begin{align}
	\vec{v} &= R^{H_{1}}_{z} q_{2}P_{0}J\vec{v}
\end{align}
taking the inner product with $\chi_0$:
\begin{align}
	(\chi_{0},\vec{v}) &= (\chi_{0}, R^{H_{1}}_{z} q_{2}P_{0}J\vec{v}) \\
	\begin{bmatrix} v_{1}(0) \\ v_{2}(0) \end{bmatrix} &= q_{2} \begin{bmatrix} 0 & f^{L}_{z_{1}} \\ f^{L}_{z_{2}} & 0 \end{bmatrix} \begin{bmatrix} v_{1}(0) \\ v_{2}(0) \end{bmatrix} .
\end{align}
One may then substitute the expression for $v_2$ into the equation for $v_1$:
\begin{align}
	v_{1}(0) &= q_{2}^{2} f^{L}_{z_{1}} f^{L}_{z_{2}} v_{1}(0) \\
	1 &= q_{2}^{2} f^{L}_{z_{1}} f^{L}_{z_{2}} .
\end{align}
\end{proof}

Next we find some preliminary estimates which are asymptotic in $\mu \nearrow \infty$.

\begin{lem}
Let $\widehat{\epsilon} := (\psi_{-\mu},(I-P_{0})\alpha^{2\sigma+1})$. It is the case that $\rho^{-1}\widehat{\epsilon} = \mu^{-(2\sigma+2)}+\mathcal{O}(a^{-(2\sigma+3)})$ for all $\sigma \in \mathbb{Z}_{+}$.
\end{lem}

\begin{proof}
One has that $(L_{0} + \mu)\alpha = \alpha^{2\sigma+1} \Rightarrow \alpha = R^{L_{0}}_{-\mu}\alpha^{2\sigma+1} \Rightarrow 1 = \rho^{2\sigma}f_{-\mu}+\rho^{-1}\widehat{\epsilon}$. Let $\epsilon_{x} := \alpha(x)/\alpha(x-1)$, for all $0 < x \in \mathbb{Z}$.
One may observe that
\begin{align}
	-\mu\alpha + \alpha^{2\sigma+1} &= L_{0}\alpha \\
	-\mu\alpha(x) + \alpha^{2\sigma+1}(x) &= -(x+1)\alpha(x+1) + (2x+1)\alpha(x) - x\alpha(x-1) .
\end{align}
By the definition of the $\epsilon_x$ one may write
\begin{align}
	\rho \epsilon_{1}\ldots\epsilon_{x-1} + (\rho\epsilon_{1}\ldots\epsilon_{x})^{2\sigma+1} &= -(x+1)\rho\epsilon_{1}\ldots\epsilon_{x+1} + (2x+1+\mu)\rho\epsilon_{1}\ldots\epsilon_{x}
\end{align}
and then by elementary algebra:
\begin{align}
	\rho + (\rho\epsilon_{1}\ldots\epsilon_{x-1})^{2\sigma}\epsilon_{x}^{2\sigma+1} &= -(x+1)\epsilon_{x}\epsilon_{x+1} + (2x+1+\mu)\epsilon_{x} \\
	\epsilon_{x} &= \left[ \mu+1+2x-(x+1)\epsilon_{x+1} \right]^{-1} \left[\rho + (\rho\epsilon_{1}\ldots\epsilon_{x-1})^{2\sigma}\epsilon_{x}^{2\sigma+1}\right] .
\end{align}
The analogous equation for $x = 0$ is
\begin{align}
	\rho^{2\sigma} = \mu +1 - \epsilon_{1}.
\end{align}

By Lemma 5.1 of \cite{paper 02} one has that $\sum_{x=0}^{\infty}\psi_{-\mu}(x) = \mu^{-1}$. Therefore
\begin{align}
	\sum_{x=1}^{\infty}\psi_{-\mu}(x) &= \sum_{x=0}^{\infty}\psi_{-\mu}(x) - \psi_{-\mu}(0) = \mu^{-1} - [\mu^{-1} - \mu^{-2} + \mathcal{O}(\mu^{-3})] \\
		&= \mu^{2} + \mathcal{O}(\mu^{-3}),
\end{align}
\begin{align}
	\rho^{-1}\widehat{\epsilon} &= \rho^{-1}\sum_{x=1}^{\infty}\psi_{-\mu}(x)\alpha^{2\sigma+1}(x) < \rho^{-1}\sum_{x=1}^{\infty}\psi_{-\mu}(x)\alpha^{2\sigma+1}(1) \\
		&= \rho^{2\sigma}\epsilon^{2\sigma+1}_{1}\sum_{x=1}^{\infty}\psi_{-\mu}(x) < (\mu+1)[\mu^{-1} + \mathcal{O}(\mu^{-2})]^{2\sigma+1}[\mu^{-2} + \mathcal{O}(\mu^{-3})] \\
		&< \mu^{-(2\sigma+2)}+\mathcal{O}(\mu^{-(2\sigma+3)}).
\end{align}
Since the lowest order of this bound has a power which depends explicitly on $\sigma$ and should hold for all $0 < \sigma \in \mathbb{Z}$ it follows that the full asymptotic expansion, as well as the error term, must consist of orders which are powers of $\mu^{-1}$ that depend explicitly on $\sigma$.
\begin{align}
	\psi_{-\mu}(1) &= (1+\mu)\psi_{-\mu}(0)-1 \\
		&= (1+\mu)[\mu^{-1}-\mu^{-2}+2\mu^{-3}+\mathcal{O}(\mu^{-4})]-1 = \mu^{-2}+\mathcal{O}(\mu^{-3}),
\end{align}
\begin{align}
	\rho^{-1}\widehat{\epsilon} &= \rho^{-1}\sum_{x=1}^{\infty}\psi_{-\mu}(x)\alpha^{2\sigma+1}(x) \\
		&= \psi_{-\mu}(1)\rho^{2\sigma}\epsilon^{2\sigma+1}_{1} + \rho^{-1}\sum_{x=2}^{\infty}\psi_{-\mu}(x)\alpha^{2\sigma+1}(x) \\
		&= [\mu^{-2}+\mathcal{O}(\mu^{-3})][\mu+1+\mathcal{O}(\mu^{-1})][\mu^{-1}+\mathcal{O}(\mu^{-2})]^{2\sigma+1} \\
			&\quad\quad + \rho^{-1}\sum_{x=2}^{\infty}\psi_{-\mu}(x)\alpha^{2\sigma+1}(x) \\
		&= \mu^{-(2\sigma+2)} + \mathcal{O}(\mu^{-(2\sigma+3)}) + \rho^{-1}\sum_{x=2}^{\infty}\psi_{-\mu}(x)\alpha^{2\sigma+1}(x)
\end{align}
From the above expansion one may conclude the bound
\begin{align}
	\rho^{-1}\sum_{x=2}^{\infty}\psi_{-\mu}(x)\alpha^{2\sigma+1}(x) = \mathcal{O}(\mu^{-(2\sigma+3)}) .
\end{align}
from which one may bound the higher order terms directly:
\begin{align}
	\rho^{-1}\widehat{\epsilon} = \mu^{-(2\sigma+2)}+\mathcal{O}(\mu^{-(2\sigma+3)})
\end{align}
\end{proof}

Now we show that there are no real eigenvalues through a series of lemmas.

\begin{lem}
It is the case that $h(0) = 2\sigma^{-1}\mu^{-(2\sigma+2)} + \mathcal{O}(\mu^{-(2\sigma+3)}) > 0$.
\end{lem}

\begin{proof}
One may observe that
\begin{align}
	\rho^{2\sigma}f_{-\mu} &= 1 - \mu^{-(2\sigma+2)} + \mathcal{O}(\mu^{-(2\sigma+3)}) \\
	f_{-\mu}^{-1} &= \mu + 1 - \mu^{-1} + 3\mu^{-2} - 16\mu^{-3} + \mathcal{O}(\mu^{-4}) \\
	\rho^{2\sigma} &= \mu + 1 - \mu^{-1} + 3\mu^{-2} - 16\mu^{-3} + \mathcal{O}(\mu^{-4}) \\
		&\quad\quad - \mu^{-(2\sigma+1)} - \mu^{-(2\sigma+2)} + \mathcal{O}(\mu^{-(2\sigma+3)}) \\
	\rho^{2\sigma}f_{-2\mu} &= 2^{-1}[1 + 2^{-1}\mu^{-1} - \mu^{-2} + \mathcal{O}(\mu^{-3})] \\
	q_{1}f_{-2\mu} -1 &= 2^{-1}(\sigma-1) + 4^{-1}(\sigma+1)\mu^{-1} - 4^{-1}(\sigma+1)\mu^{-2} + \mathcal{O}(\mu^{-3})
\end{align}

\begin{align}
	h(0) &= [(q_{1}f_{-\mu} - 1)^{-1}q_{2}f_{-\mu}]^{2} - 1 \\
		&= \{ \sigma^{-1}[1 - (1+\sigma^{-1})\mu^{-(2\sigma+2)}+\mathcal{O}(\mu^{-(2\sigma+3)})]^{-1}\sigma \rho^{2\sigma}f_{-\mu} \}^{2} - 1 \\
		&= \{ [1 + (1+\sigma^{-1})\mu^{-(2\sigma+2)}+\mathcal{O}(\mu^{-(2\sigma+3)})] \times [1-\mu^{-(2\sigma+2)}+\mathcal{O}(\mu^{-(2\sigma+3)})] \}^{2} - 1 \\
		&= [ 1 + \sigma^{-1}\mu^{-(2\sigma+2)} + \mathcal{O}(\mu^{-(2\sigma+3)}) ]^{2} - 1 \\
		&= 2\sigma^{-1}\mu^{-(2\sigma+2)} + \mathcal{O}(\mu^{-(2\sigma+3)}) > 0 .
\end{align}
\end{proof}

\begin{lem}
For $z = a \in [-\mu, \mu]$, we define $a_{1} := a - \mu$, $a_{2} := -a - \mu$. One has for $\sigma = 1$
\begin{align}
	\lim_{a \nearrow \mu}h(a) = 2^{-1}\mu - 4^{-1} + \mathcal{O}(\mu^{-1}) > 0
\end{align}
and for $1 < \sigma \in \mathbb{Z}_{+}$
\begin{align}
	\lim_{a \nearrow \mu}h(a) &= (\sigma^{2}-1)^{-1} + (\sigma+1)^{-1}(\sigma-1)^{-2}\sigma^{2}\mu^{-1} \\
		&\quad\quad + 4^{-1}(\sigma+1)^{-1}(\sigma-1)^{-3}(3\sigma^{2}+10\sigma-9)\mu^{-2} + \mathcal{O}(\mu^{-3}) > 0 .
\end{align}
\end{lem}

\begin{proof}
Let $s_{1} := -a_{1}/a$, $s_{1} \in [0,2]$. For $a_{1} \gg 1$ one has
\begin{align}
	&(\sigma+1)^{-1}s_{1}(q_{1}f_{a_{1}} - 1) = [1 - (\sigma+1)^{-1}s_{1}] + (1-s_{1}^{-1})\mu^{-1} \\
		&\quad\quad - (1 + s_{1}^{-1} - 2s_{1}^{-2})\mu^{-2} + (16 + 3s_{1}^{-1} + 2s_{1}^{-2}+6s_{1}^{-3}-24s_{1}^{-4})\mu^{-4} \\
		&\quad\quad + \mathcal{O}(\mu^{-5}) - \mu^{-(2\sigma+2)} - (1-s_{1}^{-1})\mu^{-(2\sigma+3)} + \mathcal{O}(\mu^{-(2\sigma+4)}) .
\end{align}

The generalized exponential integrals have the convergent series expansion \cite{En}
\begin{align}
	E_{n+1}(z) &= -\frac{(-z)^n}{n!} \log(z) + \frac{e^{-z}}{n!}\sum_{k=1}^n(-z)^{k-1}(n-k)!  + \frac{e^{-z}(-z)^n}{n!}\sum_{k=0}^\infty \frac{z^k}{n!}\digamma(k+1) ,\label{convergent}
\end{align}
where $\digamma(x) := \mathrm{d}_x \log \Gamma(x)$ is the digamma function. This permits one to write
\begin{align}
	f_{-a} &= e^{a}E_{1}(a) = -e^{a}\log a + \sum_{k = 0}^{\infty}{a^{k} \over k!}\digamma(k+1)
\end{align}
and hereby one may observe that near $a = 0$ the above expression for $f_{-a}$ is dominated by logarithmic behavior. Therefore $\lim_{a \nearrow \mu} (-f^{L}_{a_{1}}) = \lim_{a \nearrow \mu} (q_{1}f_{a_{1}}-1)^{-1}f_{a_{1}} = q_{1}^{-1}$.

\begin{align}
	\lim_{a \nearrow \mu}h(a) &= \lim_{a \nearrow \mu}(q_{1}f_{z_{1}} - 1)^{-1}q_{2}f_{z_{1}}(q_{1}f_{z_{2}} - 1)^{-1}q_{2}f_{z_{2}} - 1\\
		&= q_{1}^{-1}q_{2}^{2}(q_{1}f_{-2\mu}-1)^{-1}f_{-2\mu} - 1 \\
		&= (\sigma+1)^{-1}\sigma^{2}[(\sigma+1)\rho^{2\sigma}f_{-2\mu} -1]^{-1}\rho^{2\sigma}f_{-2\mu}-1 \\
		&= 2^{-1}(\sigma+1)^{-1}\sigma^{2}[2^{-1}(\sigma-1) + 4^{-1}(\sigma+1)\mu^{-1} \\
			&\quad\quad -4^{-1}(\sigma+1)\mu^{-2}+\mathcal{O}(\mu^{-3})]^{-1} [1+2^{-1}\mu^{-1}-\mu^{-2}+\mathcal{O}(\mu^{-3})]
\end{align}
\underline{For $\sigma=1$}:
\begin{align}
	\lim_{a \nearrow \mu}h(a) &= 4^{-1}[2^{-1}\mu^{-1}-2^{-1}\mu^{-2}+\mathcal{O}(\mu^{-3})]^{-1}[1+2^{-1}\mu^{-1}-\mu^{-2}+\mathcal{O}(\mu^{-3})]-1 \\
		&= 2^{-1}\mu[1-\mu^{-1}+\mathcal{O}(\mu^{-2})]^{-1}[1+2^{-1}\mu^{-1}-\mu^{-2}+\mathcal{O}(\mu^{-3})]-1 \\
		&= 2^{-1}\mu[1+\mu^{-1}+\mathcal{O}(\mu^{-2})][1+2^{-1}\mu^{-1}-\mu^{-2}+\mathcal{O}(\mu^{-3})]-1 \\
		&= 2^{-1}\mu-4^{-1}+\mathcal{O}(\mu^{-1}) > 0
\end{align}
\underline{For $\sigma>1$}:
\begin{align}
	\lim_{a \nearrow \mu}h(a) &= (\sigma^{2}-1)^{-1}\sigma^{2}\{ 1-2^{-1}(\sigma-1)^{-1}(\sigma+1)[-\mu^{-1}+\mu^{-2}+\mathcal{O}(\mu^{-3})] \}^{-1} \\
			&\quad\quad \times[1+2^{-1}\mu^{-1}-\mu^{-2}+\mathcal{O}(\mu^{-3})] - 1 \\
		&= (\sigma^{2}-1)^{-1}\sigma^{2}\{ 1 + 2^{-1}(\sigma-1)^{-1}(\sigma+1)[-\mu^{-1}+\mu^{-2}+\mathcal{O}(\mu^{-3}) \\
			&\quad\quad + [2^{-1}(\sigma-1)^{-1}(\sigma+1)]^{2}[-\mu^{-1}+\mathcal{O}(\mu^{-2})]^{2} + \mathcal{O}(\mu^{-3}) \} \\
			&\quad\quad \times[1+2^{-1}\mu^{-1}-\mu^{-2}+\mathcal{O}(\mu^{-3})] - 1 \\
		&= (\sigma^{2}-1)^{-1}\sigma^{2}[1-2^{-1}(\sigma-1)^{-1}(\sigma+1)\mu^{-1} \\
			&\quad\quad +4^{-1}(\sigma-1)^{2}(\sigma+1)(3\sigma-1)\mu^{-2}+\mathcal{O}(\mu^{-3})] \\
			&\quad\quad \times[1+2^{-1}\mu^{-1}-\mu^{-2}+\mathcal{O}(\mu^{-3})] - 1 \\
		&= (\sigma^{2}-1)^{-1}\sigma^{2}[1-(\sigma-1)^{-1}\mu^{-1} \\
			&\quad\quad+4^{-1}(\sigma-1)^{2}(3\sigma^{2}10\sigma-9)\mu^{-2}+\mathcal{O}(\mu^{-3})] -1 \\
		&= (\sigma^{2}-1)^{-1}-(\sigma+1)^{-1}(\sigma-1)^{-2}\sigma^{2}\mu^{-1} \\
			&\quad\quad + 4^{-1}(\sigma+1)^{-1}(\sigma-1)^{3}(3\sigma^{2}+10\sigma-9)\mu^{-2} + \mathcal{O}(\mu^{-3}) > 0
\end{align}
\end{proof}

\begin{lem}
Let $h_{0}(z) := (q_{1}f_{z_{1}}-1)^{-1}(q_{1}f_{z_{2}}-1)^{-1}(2\sigma+1)\rho^{4\sigma}$. It is the case that $h_{0}(z) > 0$ for all $z = a \in [-\mu, \mu]$.
\end{lem}

\begin{proof}
Let $c_{0} := (2\sigma+1)^{-1}\rho^{-4\sigma}$, $c_{1} := (2\sigma+1)^{-1}(\sigma+1)\rho^{-2\sigma}$, $c_{2} := c_{1}^{2} - c_{0} = (2\sigma+1)^{-2}\rho^{-4\sigma}\sigma^{2}$. One then has $h(z) = h_{0}(z)[c_{2} - (f_{z_{1}} - c_{1})(f_{z_{2}} - c_{1})]$. For all $\sigma \in \mathbb{Z}_{+}$ and $a \in [-\mu, \mu]$ it is the case that $q_{1}f_{a_{1}} = q_{1}f_{a - \mu}$ and $q_{1}f_{a_{2}} = q_{1}f_{-a - \mu}$ vary monotonically between
\begin{align}
	q_{1}f_{-2\mu} = 2^{-1}(\sigma+1) + 4^{-1}(\sigma+1)\mu^{-1} - 4^{-1}(\sigma+1)\mu^{-2} + \mathcal{O}(\mu^{-3}) > 1
\end{align}
and
\begin{align}
	\lim_{a \nearrow 0}q_{1}f_{-a} = \infty > 1.
\end{align}
Therefore $(q_{1}f_{a_{i}}-1)^{-1} > 0$ which concludes the proof.
\end{proof}

\begin{lem}
For $1 \ll (2\sigma+1)^{-1}\mu = \mathcal{O}(\mu)$ it is the case that $(f_{a_{j}} - c_{1})$, $j=1,2$, have unique roots, $a = r_{j} := (-1)^{j}(\sigma+1)^{-1}(\mu+\sigma) + \mathcal{O}(\mu^{-1})$.
\end{lem}

\begin{proof}
Consider $c_{1} = f_{a_{j}}$. Let $a_{1} = -s_{1}\mu$, $s_{1} \in [0,2]$, and assume that $1 \ll s_{1}\mu = \mathcal{O}(\mu)$ so that the asymptotic expansion of $f_{a_{1}}$ is valid.
\begin{align}
	c_{1} &= f_{a_{1}} = (s_{1}\mu)^{-1} - (s_{1}\mu)^{-2} + \mathcal{O}(\mu^{-3}) .
\end{align}
One may multiply both sides by $s_1/c_1$ and expand iteratively in $s_1$:
\begin{align}
	 s_{1} &= \mu^{-1}c_{1}^{-1}[1 - (s_{1}\mu)^{-1} + \mathcal{O}(\mu^{-2})] \\
		&= (\sigma+1)^{-1}(2\sigma+1)[1+\mu^{-1}+\mathcal{O}(\mu^{-2})] \times [1 - (s_{1}\mu)^{-1} + \mathcal{O}(\mu^{-2})] \\
		&= (\sigma+1)^{-1}(2\sigma+1)[1+(1-s^{-1})\mu^{-1}+\mathcal{O}(\mu^{-2})] \\
		&= (\sigma+1)^{-1}(2\sigma+1) \times \{1+[1-(2\sigma+1)^{-1}(\sigma+1)]\mu^{-1}+\mathcal{O}(\mu^{-2})\} \\
		&= (\sigma+1)^{-1}(2\sigma+1) + (\sigma+1)^{-1}\sigma\mu^{-1} + \mathcal{O}(\mu^{-2}).
\end{align}
This result satisfies the assumptions and the root must be unique, therefore the above value of $s_{1}$ specifies the unique root. In terms of $a$ one has
\begin{align}
	a_{1} &= a - \mu = -s_{1}\mu ,
\end{align}
which implies that
\begin{align}
	a &= (1-s_{1})\mu = \{ 1 - [(\sigma+1)^{-1}(2\sigma+1) + (\sigma+1)^{-1}\sigma\mu^{-1} + \mathcal{O}(\mu^{-2})] \}\mu \\
		&= -(\sigma+1)^{-1}(\mu+\sigma)+\mathcal{O}(\mu^{-1}).
\end{align}
Due to symmetry between the $a_{j}$ there are two roots $a = r_{j} = (-1)^{j}(\sigma+1)^{-1}(\mu+\sigma) + \mathcal{O}(\mu^{-1})$.
\end{proof}

\begin{lem}
For $1 \ll r_{j} = \mathcal{O}(\mu)$ one has that $h(z) > 0$ for all $z = a \in [-\mu, \mu]$.
\end{lem}

\begin{proof}
Let $h_{1}(z) := (f_{z_{1}} - c_{1})(f_{z_{2}} - c_{1})$. One can observe that $h_{1}(0) = (f_{-\mu} - c_{1})^{2} > 0$ and $\lim_{a \nearrow \mu}(f_{a_{1}} - c_{1}) = \infty$. Then since $(f_{a_{j}} - c_{1})$ are monotonic in $a$ and have unique roots $r_{i}$, it is the case that $h(a) > 0$ for $a \in (r_{1},r_{2})$ and $h(a) < 0$ for $a \in [-\mu,\mu] \setminus (r_{1},r_{2})$. We have shown that $h(0) > 0$ and $h(\mu) = h(-\mu) > 0$. Since $h(a) > 0$ for $a \in (r_{1},r_{2})$ and $h(a) < 0$ for $a \in [-\mu,\mu] \setminus (r_{1},r_{2})$ it must be the case that $h(a) > 0$ for $a \in [-\mu,\mu] \setminus (r_{1},r_{2})$. It remains to consider $a \in (r_{1},r_{2})$. By assumption $1 \ll a = \mathcal{O}(\mu)$ and therefore the asymptotic expansion of the $f_{a_{j}}$ is valid.
\begin{align}
	h_{1}(a) &= (f_{a_{1}} - c_{1})(f_{a_{2}} - c_{1}) = f_{a_{1}}f_{a_{2}} - c_{1}(f_{a_{1}}+f_{a_{2}}) + c_{1}^{2} \\
		&= [(\mu-a)^{-1} + \mathcal{O}(\mu^{-2})][(\mu+a)^{-1} + \mathcal{O}(\mu^{-2})] \\
			&\quad\quad -c_{1}\{ [(\mu-a)^{-1} + \mathcal{O}(\mu^{-2})] + [(\mu+a)^{-1} + \mathcal{O}(\mu^{-2})] \} + c_{1}^{2} \\
		&= (\mu^{2}-a^{2})^{-1} - (\mu^{2}-a^{2})^{-1}(2\sigma+1)^{-1}(\sigma+1) \\
			&\quad\quad \times[1+\mathcal{O}(\mu^{-1})][2+\mathcal{O}(\mu^{-1})] + c_{1}^{2} + \mathcal{O}(\mu^{-3}) \\
		&= -(\mu^{2}-a^{2})^{-1}(2\sigma+1)^{-1}+c_{1}^{2} + \mathcal{O}(\mu^{-3}) .
\end{align}
Therefore $-h(a)$ decays monotonically as $|a| \nearrow r_{2}$ for sufficiently large $\mu$. This guarantees that $h(a) > 0$ for all $a \in (r_{1},r_{2})$.
\end{proof}

\begin{lem}
It is the case that $h(z)$ has no roots for $z \in (-\infty, -\mu] \cup [\mu, \infty)$.
\end{lem}

\begin{proof}
This follows by the same principle which forbids $L$ from having embedded eigenvalues.
\end{proof}

Now we prove the existence and location (asymptotically) of the imaginary roots.

\begin{lem}
$h(z)$ has exactly two roots, $\lambda_{\pm} = z_{\pm} := \pm i (2\sigma)^{1/2}\mu^{-\sigma}[1+\mathcal{O}(\mu^{-1})]$, for $z \in i\mathbb{R}$.
\end{lem}

\begin{proof}
One may observe that
\begin{align}
	f_{z} &= (-z)^{-1} - (-z)^{-2} + 2(-z)^{-3} - 6(-z)^{-4} + \mathcal{O}(z^{-5}), \\
	f^{L}_{z} &= (1-q_{1}f_{z})^{-1}f_{z},\quad \partial_{z}f^{L}_{z} = (1-q_{1}f_{z})^{-2}\partial_{z}f_{z} \\
	\partial_{z}f_{z} &= (-1)(f_{z}+z^{-1}),\quad \partial^{2}_{z}f_{z} = 2q_{1}(1-q_{1}f_{z})^{-3}(\partial_{z}f_{z})^{2} + (1-q_{1}f_{z})^{-2}\partial^{2}_{z}f_{z}, \\
	f'_{-\mu} &= \mu^{-2} -2\mu^{-3}+6\mu^{-4}+\mathcal{O}(\mu^{-5}),\quad f''_{-\mu} = 2\mu^{-3}+6\mu^{-4}+\mathcal{O}(\mu^{-5})
\end{align}
and that for $z \in i\mathbb{R}$ one has that $h(z) = |q_{2}f^{L}_{\mu-z}|^{2} -1 \in \mathbb{R}$. Furthermore, for all $z \in i\mathbb{R}$ the asymptotic expansion of $f_{z_{j}}$ is valid since $\Re z_{j} = a \gg 1$.

First consider $|z| \ll 1$. One finds
\begin{align}
	\partial^{2}_{z}h(z) &= q^{2}_{2}(\partial^{2}_{z}f^{L}_{z_{1}}f^{L}_{z_{2}} + 2\partial_{z}f^{L}_{z_{1}}\partial_{z}f^{L}_{z_{2}} + f^{L}_{z_{1}}\partial^{2}_{z}f^{L}_{z_{2}}) \\
		&= q^{2}_{2}(\partial^{2}_{z_{1}}f^{L}_{z_{1}}f^{L}_{z_{2}} - 2\partial_{z_{1}}f^{L}_{z_{1}}\partial_{z_{2}}f^{L}_{z_{2}} + f^{L}_{z_{1}}\partial^{2}_{z_{2}}f^{L}_{z_{2}}) \\
		&= q^{2}_{2}\{ [2q_{1}(1-q_{1}f_{z_{1}})^{-3}(\partial_{z_{1}}f_{z_{1}})^{2} + (1-q_{1}f_{z_{1}})^{-2}\partial^{2}_{z_{1}}f_{z_{1}}] \} (1-q_{1}f_{z_{2}})^{-1}f_{z_{2}} \\
			&\quad\quad - 2(1-q_{1}f_{z_{1}})^{-2}\partial_{z_{1}}f_{z_{1}}(1-q_{1}f_{z_{2}})^{-2}\partial_{z_{2}}f_{z_{2}} \\
			&\quad\quad + (1-q_{1}f_{z_{1}})^{-1}f_{z_{1}} [2q_{1}(1-q_{1}f_{z_{2}})^{-3}(\partial_{z_{2}}f_{z_{2}})^{2} + (1-q_{1}f_{z_{2}})^{-2}\partial^{2}_{z_{2}}f_{z_{2}}] \} \\
	h''(0) &= q^{2}_{2}\{ [2q_{1}(1-q_{1}f_{-\mu})^{-3}(f'_{-\mu})^{2} + (1-q_{1}f_{-\mu})^{-2}f''_{-\mu}] \} (1-q_{1}f_{-\mu})^{-1}f_{-\mu} \\
			&\quad\quad - 2(1-q_{1}f_{-\mu})^{-2}f'_{-\mu}(1-q_{1}f_{-\mu})^{-2}f'_{-\mu} \\
			&\quad\quad + (1-q_{1}f_{-\mu})^{-1}f_{-\mu} [2q_{1}(1-q_{1}f_{-\mu})^{-3}(f'_{-\mu})^{2} + (1-q_{1}f_{-\mu})^{-2}f''_{-\mu}] \}
\end{align}
\begin{align}
		&= q^{2}_{2}[2(q_{1}f_{-\mu}-1)^{-4}(f'_{-\mu})^{2}(2q_{1}f_{-\mu}-1) - 2(q_{1}f_{-\mu}-1)^{-3}f''_{-\mu}f_{-\mu}] \\
		&= q^{2}_{2}( 2\{\sigma[1+(1+\sigma^{-1})\mu^{-(2\sigma+2)}+\mathcal{O}(\mu^{-(2\sigma+3)})]\}^{-4} \\
			&\quad\quad \times [\mu^{-2}-2\mu^{-3}+6\mu^{-4}+\mathcal{O}(\mu^{-6})]^{2} \\
			&\quad\quad \times \{ 2[(\sigma+1)-(\sigma+1)\mu^{-(2\sigma+2)}+\mathcal{O}(\mu^{-(2\sigma+3)})]-1 \} \\
			&\quad\quad - 2\{\sigma[1+(1+\sigma^{-1})\mu^{-(2\sigma+2)}+\mathcal{O}(\mu^{-(2\sigma+3)})]\}^{-3} \\
			&\quad\quad \times [2\mu^{-3}-6\mu^{-4}+\mathcal{O}(\mu^{-5})] \\
			&\quad\quad \times [\mu^{-1}-\mu^{-2}+2\mu^{-3}-6\mu^{-4}+\mathcal{O}(\mu^{-5})] ) \\
		&= 2\sigma^{-2}\mu^{-2}[1-2\mu^{-1}+\mathcal{O}(\mu^{-2})] .
\end{align}
Assume that $h(z) = h(0) + 2^{-1}h''(0)z^{2} + \epsilon$ where $|\epsilon| \le \mathcal{O}(\mu^{-4})$. One finds that $h(z) = 0$ implies:
\begin{align}
	z^{2} &= -2[h(0) + \epsilon][h''(0)]^{-1} \\
		&= -(2)[2\sigma^{-1}\mu^{-(2\sigma+2)}+\mathcal{O}(\mu^{-(2\sigma+3)})+\epsilon] (2^{-1}\sigma^{2}\mu^{2})[1-2\mu^{-1}+\mathcal{O}(\mu^{-2})]^{-1} \\
		&= -2\sigma\mu^{2\sigma}[1+\mathcal{O}(\mu^{-1})+\epsilon].
\end{align}
This result is compatible with the assumptions, thus there are at least the two imaginary roots given by $z_{\pm} = \pm i (2\sigma)^{1/2}\mu^{-\sigma}[1+\mathcal{O}(\mu^{-1})]$.

It remains to be shown that there are no other imaginary roots. It is sufficient to prove that $h(z)$ has nonpositive curvature for all $z \in i\mathbb{R}$. Let $z_{j} = (-1)^{j+1}ib - \mu$, where $b \in \mathbb{R}$.
\begin{align}
	h(z) &= q_{2}^{2}(q_{1}-f_{z_{1}}^{-1})(q_{1}-f_{z_{2}}^{-1}) - 1 \\
		&= q_{2}^{2}[(\sigma+1)\mu - (\mu-ib) + \mathcal{O}(1)]^{-1}[(\sigma+1)\mu - (\mu+ib) + \mathcal{O}(1)]^{-1} - 1 \\
		&= q_{2}^{2}(\sigma^{2}\mu^{2}+b^{2})^{-1}[1+\mathcal{O}(\mu^{-1})] - 1.
\end{align}
Therefore $h(ib)$ decays monotonically as $|b| \nearrow \infty$ and there are only two imaginary roots.
\end{proof}

\begin{proof}[Proof of Theorem \ref{snthm05} Part (1)]
We have exhaustively shown that
\begin{align}
	\lambda_{\pm} := \pm i (2\sigma)^{1/2}\mu^{-\sigma}[1+\mathcal{O}(\mu^{-1})]
\end{align}
are the only roots of $h(z)$ for $z \in \mathbb{R} \cup i\mathbb{R}$. The absence of embedded eigenvalues follows from arguments similar to those for $\sigma(L)$.
\end{proof}

\begin{proof}[Proof of Theorem \ref{snthm05} Part (2)]
By Weyl's critereon it is the case that $\sigma_{\mathrm{e}}(H_{2}) = \sigma_{\mathrm{e}}(H_{0}) = (-\infty, -\mu] \cup [\mu, \infty)$. It is clear that there exists a well-defined absolutely continuous spectral measure on $\sigma_{\mathrm{e}}(H_{2})$. The representation of $H_{2}v = z v$ as a coupled series of algebraic equations guarantees that each $\lambda \in \sigma_{\mathrm{e}}(H_{2})$ has multiplicity 1. Therefore one must have that $\sigma_{\mathrm{e}}(H_{2}) = \sigma_{\mathrm{ac}}(H_{0})$.
\end{proof}

\section{Spectral Properties of $H$}

We consider without proof Proposition 1 of \cite{paper 02}:
\begin{align}
	||(I-P_{0})\alpha_{\mu}||_{1} \le \mu^{-(2\sigma)^{-1}(2\sigma-1)} + \mathcal{O}(\mu^{-(2\sigma)^{-1}(4\sigma-1)}) .
\end{align}
This gives
\begin{align}
	||U|| &\le 2(2\sigma + 1)||(I - P_{0})\alpha^{2\sigma}||_{1} \le 2(2\sigma + 1)||(I - P_{0})\alpha||_{1} \\
		&\le 2(2\sigma + 1)\mu^{-(2\sigma)^{-1}(2\sigma-1)} + \mathcal{O}(\mu^{-(2\sigma)^{-1}(4\sigma-1)}) =: m(\mu) .
\end{align}
We recall without proof a proposition of Kato \cite{Kato} regarding norm resolvent convergence.
\begin{unprop}
For $A$ a closed operator and $\{ A_{n} \}_{n=0}^{\infty}$ a sequence of closed operators, if $R^{A_{n}}_{z}$ converges in norm to $R^{A}_{z}$ for some $z \in \rho(A)$ then the convergence holds for every $z \in \rho(A)$.
\end{unprop}

\begin{proof}[Proof of Theorem \ref{snthm06} Part (1)]
The discrete spectrum of $H$ can be at most $||U|| \le m(\mu)$ away from that of $H_{2}$. We therefore only need to consider the shift of the eigenvalues of $H_{2}$ and possible production of eigenvalues from the thresholds of $H_{2}$.

Consider the eigenvalues near the origin. By standard arguments, see e.g. \cite{Schlag}, the kernel of $H$ is spanned by linear combinations of matrix vectors composed the set $\{ T_{j}\alpha \}_{j = 1}^{n}$ where $\{T_{j}\}_{j=1}^{n}$ is the set of generators of symmetries of the soliton manifold, in which $\alpha$ lies. In our case there are only two symmetries: phase rotation and energy translation. The kernel is then spanned by matrix linear combinations of $\alpha$ and $\partial_{\mu}\alpha$ and thereby there exists an eigenvalue of multiplicity 2 at the origin. These eigenvalues must be result of the shift of the eigenvalues of $H_{2}$ to the origin.

Now consider the possibility of eigenvalues near the threshold. Consider that by the resolvent identity, one has $R^{H}_{z} - R^{H_{2}}_{z} = R^{H_{2}}_{z}UR^{H}_{z}$. Without loss of generality, let $z$ be chosen so that $|| R^{H_{2}}_{z} || \le |z + \epsilon(\mu)|^{-1}$ and $|z+\epsilon(\mu)|^{-1}||U|| < 1$ for some $\epsilon(\mu) = \mathcal{O}(\mu^{-\sigma})$, i.e. $\epsilon(\mu)$ is due to the presence of the eigenvalues of $H_{2}$. Then
\begin{align}
	\lim_{\mu \nearrow \infty}||R^{H_{2}}_{z}|| \le \lim_{\mu \nearrow \infty} |z + \epsilon(\mu)|^{-1} = |z|^{-1}
\end{align}
and
\begin{align}
	||R^{H}_{z}|| &= ||(1-R^{H_{2}}_{z}U)^{-1}R^{H_{2}}_{z}|| \le (1-||R^{H_{2}}_{z}||\ ||U||)^{-1}||R^{H_{2}}_{z}|| \\
		&\le [1 - |z + \epsilon(\mu)|^{-1}m(\mu)]^{-1} |z + \epsilon(\mu)|^{-1} \\
	\lim_{\mu \nearrow \infty} ||R^{H}_{z}|| &\le \lim_{\mu \nearrow \infty} [1 - |z + \epsilon(\mu)|^{-1}m(\mu)]^{-1} |z + \epsilon(\mu)|^{-1} = |z|^{-1}.
\end{align}
One may then find
\begin{align}
	\lim_{\mu \nearrow \infty} || R^{H}_{z} - R^{H_{2}}_{z} || &= \lim_{\mu \nearrow \infty} || R^{H_{2}}_{z}UR^{H}_{z} || \le \lim_{\mu \nearrow \infty} || R^{H_{2}}_{z}||\ ||U||\ ||R^{H}_{z} || \\
		&\le \lim_{\mu \nearrow \infty} |z|^{-2}m(\mu) = 0.
\end{align}
It is the case that $A$ is closed if $R^{A}_{z}$ exists and is bounded for at least one $z \in \mathbb{C}$. This is clearly the case for both $H_{2}$, $H$. Therefore by the principle of norm resolvent convergence it is the case that $(-\mu, 0) \cup (0,\mu) \subset \rho(H)$.
\end{proof}

\begin{proof}[Proof of Theorem \ref{snthm06} Part (2)]
By Weyl's critereon $\sigma_{\mathrm{e}}(H) = \sigma_{\mathrm{e}}(H_{2})$. One may explicitly construct an absolutely continuous spectral measure by expanding $R^{H}_{z} = (1-R^{H_{2}}_{z}U)^{-1}R^{H_{2}}_{z}$ as a convergent series in $U$, taking a limit $z \to \lambda \in \sigma_{\mathrm{e}}(H_{2})$, and collecting the imaginary terms. The representation of $Hv = z v$ as a coupled series of algebraic equations guarantees that each $\lambda \in \sigma_{\mathrm{e}}(H)$ has multiplicity 1. Therefore one must have that $\sigma_{\mathrm{e}}(H) = \sigma_{\mathrm{ac}}(H_{2})$.
\end{proof}

\section{Decay Estimates for $H_{2}$ and $H$}

\begin{defn}
For any single-valued or multi-valued function $f: \mathbb{C} \to \mathbb{C}$, an element of a set of linear functionals on some suitable Banach space with norm given through integration over $\lambda$, and with poles, branch points, and branch cuts found in the subset $\Sigma \subseteq \mathbb{R}$ let $\mathcal{PV}f : \Sigma \to \mathbb{C}$ be the \emph{principal value of $f$} defined by the weak limit
\begin{align}
	\mathcal{PV}f(\lambda) :=& \frac{1}{2} \wlim_{\epsilon \searrow 0} \left[  f(\lambda + i\epsilon) + f(\lambda - i\epsilon) \right] , \quad \lambda \in \Sigma,
\end{align}
which converges in the distributional sense. We analogously define the \emph{$\delta$-part} of $f$ to be
\begin{align}
	\delta f(\lambda) :=& {1 \over 2\pi i} \wlim_{\epsilon \searrow 0} \left[  f(\lambda + i\epsilon) - f(\lambda - i\epsilon) \right] , \quad \lambda \in \Sigma.
\end{align}
\end{defn}

We have kept vague the specification of the sense in which the above definitions converge weakly for the purposes of generality. The details of such convergence in our work will be clear from context. One may extend the domain of $\mathcal{PV}f$ to the complex plane and produce a single valued function, which we will also denote $f$, through
\begin{align}
	\mathcal{PV}f(z) := \left\{
	\begin{array}{cr}
		 f(z),&\ z \in \mathbb{C} \setminus \Sigma \\
		\mathcal{PV}f(z),&\ z \in \Sigma .
	\end{array} \right. .
\end{align}
One may observe that the analogous extension of $\delta f(\lambda)$ vanishes away from $\Sigma \subseteq \mathbb{R}$. This prescription extends to weak limits in $z \in \mathbb{C}$ of complex sequences $v_{z} \in \mathscr{T}$ whose components depend upon $z$.

From the convergent series expansion of the generalized exponential integrals \eqref{convergent}, one has that
\begin{align}
	\wlim_{\epsilon \searrow 0} E_{n+1}(-x \pm i \epsilon) = \mathcal{PV}E_{n+1}(-x) \mp i\pi \frac{(x)^n}{n!} ,\quad x > 0,
\end{align}
where for the sake of generality the limit is weak with respect to $L^{2}([a,\infty),\mathbb{C})$, $a > 0$.
One may write $\mathcal{PV}E_1(-x) = -\mathrm{Ei}(x)$ where
\begin{align}
	\mathrm{Ei}(x) := - \int_{-x}^\infty \!\mathrm{d}u\ u^{-1}e^{-u},\qquad x>0
\end{align}
is \emph{the exponential integral}.

\begin{defn}
Consider an operator $A$ on $\mathscr{H}$ which is self-adjoint on its domain $\mathcal{D}(A)$ and $\lambda$ an element of the discrete spectrum of $A$. Define $\mathcal{PV}^A_\lambda \equiv \mathcal{PV}(A - \lambda)^{-1}$, $\lambda \in \sigma(A)$ to be the \emph{principal value of the resolvent of $A$} given by the strong limit
\begin{align}
	\mathcal{PV}^A_\lambda :=& \frac{1}{2}\slim_{\epsilon \searrow 0} \left[  R^A_{\lambda + i\epsilon} + R^A_{\lambda - i\epsilon} \right] .
\end{align}
Denote by $ \delta^A_\lambda \equiv \delta(A - \lambda) \equiv P^A_\lambda$, $\lambda \in \sigma(A)$ the spectral projection defined by the strong limit
\begin{align}
	\delta^A_\lambda :=& \frac{1}{2\pi i} \slim_{\epsilon \searrow 0} \left[  R^A_{\lambda + i\epsilon} - R^A_{\lambda - i\epsilon} \right] .
\end{align}
If $\lambda$ is instead an element of the essential spectrum of $A$ one has that $\mathcal{PV}^A_\lambda, \delta^A_\lambda$ are defined by weak limits. If and only if the essential spectrum of $A$ is absolutely continuous then it is the case that $\mathrm{d}\mu^A_{\mathrm{e}}(\lambda) = \delta^A_\lambda\ \mathrm{d}\lambda$, where $\mathrm{d}\mu^A_{\mathrm{e}}(\lambda)$ is the essential spectral measure of $A$ and $\mathrm{d}\lambda$ is the Lebesgue measure on $\sigma_{\mathrm{e}}(A)$.
\end{defn}

The above definition permits the useful representation $\delta^{A}_{\lambda} = w^{A}_{\lambda}\phi^{A}_{\lambda} \otimes \phi^{A,*}_{\lambda}$ for $A$ of generalized multiplicity 1. One may observe through the spectral representation of $R^{L_{0}}_{z}$ that $\mathcal{PV}\psi^{L_{0}}_\lambda = \mathcal{PV}^{L_{0}}_\lambda \chi_{0}$ and that $\mathcal{PV}\xi^{L_{0}}_\lambda = \mathcal{PV}\psi^{L_{0}}_\lambda - \mathcal{PV}\psi^{L_{0}}_\lambda(0)\phi^{L_{0}}_\lambda = \xi^{L_{0}}_\lambda$, $\forall \lambda \in \sigma(L_{0})$, and analogously so for other operators.

We recall the method of spectral shifts as applied to rank-1 perturbations, see e.g. \cite{rank one}. Let $A = A_{0} - qP$ where $A_{0}$ is densely defined on $\mathscr{H}$ with nonempty resolvent set and $P = v_{\mathrm{P}} \otimes v_{\mathrm{P}}^{*}$ is a rank-1 orthogonal projection. Through the resolvent formula it follows by direct algebra that
\begin{align}
	R^{A}_z &= R^{A_0}_z + R^{A_0}_zqP R^{A}_z, \quad P R^{A}_z = P R^{A_0}_z + f^{A_0}_zqP R^{A}_z  ,\\
	P R^{A}_z &=  (1 - qf^{A_0}_z)^{-1}P R^{A_0}_z, \quad R^{A}_z = R^{A_0}_z + (1-qf^{A_{0}}_{z})^{-1}R^{A_0}_zqP R^{A_0}_z .
\end{align}
For $A$ essentially self-adjoint one may apply the definitions of $\mathcal{PV}^A_\lambda$ and $\delta^A_\lambda$ and find the corresponding shifts to $\mathcal{PV}^{A_{0}}_\lambda$ and $\delta^{A_{0}}_\lambda$. For $\lambda \in \sigma(A)$ it follows that
\begin{align}
	\mathcal{PV}^A_\lambda &= \mathcal{PV}^{A_0}_\lambda + g^{A_{0}}_{\lambda}[ (1-q\mathcal{PV}f^{A_{0}}_{\lambda})(\mathcal{PV}^{A_0}_\lambda qP \mathcal{PV}^{A_0}_\lambda - \pi^2 \delta^{A_0}_\lambda qP \delta^{A_0}_\lambda) \\
		&\quad\quad - \pi^2 q\delta f^{A_{0}}_{\lambda} (\mathcal{PV}^{A_0}_\lambda qP \delta^{A_0}_\lambda + \delta^{A_0}_\lambda qP \mathcal{PV}^{A_0}_\lambda) ]
\end{align}
\begin{align}
	\delta^A_\lambda &= \delta^{A_0}_\lambda + g^{A_{0}}_{\lambda} [ (1-q\mathcal{PV}f^{A_{0}}_{\lambda})(\mathcal{PV}^{A_0}_\lambda qP \delta^{A_0}_\lambda + \delta^{A_0}_\lambda qP \mathcal{PV}^{A_0}_\lambda) \\
		&\quad\quad + q\delta f^{A_{0}}_{\lambda} (\mathcal{PV}^{A_0}_\lambda qP \mathcal{PV}^{A_0}_\lambda - \pi^2 \delta^{A_0}_\lambda qP \delta^{A_0}_\lambda) ] ,
\end{align}
where
\begin{align}
	g^{A_{0}}_{\lambda} := [ (1-q\mathcal{PV}f^{A_{0}}_{\lambda})^2 + (q \pi \delta f^{A_{0}}_{\lambda})^2]^{-1} ,\quad f^{A_{0}}_{z} := (v_{\mathrm{P}},R^{A_{0}}_{z}v_{\mathrm{P}}).
\end{align}

We recall without proof Lemma 3.12 from \cite{2D}:
\begin{unlem}
Let $\mathscr{B}$ be a Banach space and $\lambda_{+} > \lambda_{-}$ be real constants. If $F(\lambda)$ has the properties
\begin{enumerate}
\item $F \in C(\lambda_{-},\lambda_{+};\mathscr{B})$
\item $F(\lambda_{-}) = F(\lambda) = 0 , \quad \lambda > \lambda_{+}$
\item $\mathrm{d}_{\lambda} F \in L^1(\lambda_{-} + \delta, \lambda_{+};\mathscr{B}) , \quad \forall \delta>0$
\item $\mathrm{d}_{\lambda} F(\lambda) = \mathcal{O}( [\lambda - \lambda_{-}]^{-1}\log^{-3}[\lambda - \lambda_{-}] ) , \quad \lambda \searrow \lambda_{-}$
\item $\mathrm{d}_{\lambda}^{2} F(\lambda) = \mathcal{O}( [\lambda-\lambda_{-}]^{-2}\log^{-2}[\lambda - \lambda_{-}] ) , \quad \lambda \searrow \lambda_{-}$
\end{enumerate}
then
\begin{align}
\int_{\lambda_{-}}^\infty \!\mathrm{d}\lambda\ e^{-it\lambda} F(\lambda) = \mathcal{O}(t^{-1}\log^{-2}t),\quad t\nearrow\infty
\end{align}
in the norm of $\mathscr{B}$.
\end{unlem}

We will verify that $F(\lambda) = \delta^{H_{2}}_{\lambda}$ satisfies the desired properties for both $\lambda \ge \mu$ and $\lambda \le \mu$.

\begin{proof}[Proof of Theorem \ref{snthm07}]
Let $\mathscr{B} = \left\{ A \in \mathcal{L}(\vec{\mathscr{T}}) : || A ||_{\mathscr{B}} < \infty \right\} $ be the Banach space complete in the norm
\begin{align}
|| A ||_{\mathscr{B}} := \sup_{v \in \ell^{1}} { || W_{\kappa,\tau}AW_{\kappa,\tau}v ||_{1} \over || v ||_{1} } ,
\end{align}
where $\vec{\mathscr{T}} = \mathscr{T} \oplus \mathscr{T}$ is the natural extension of $\mathscr{T}$ to the matrix system. Let $F(\lambda) = \delta^{H_{2}}_{\lambda}$. We will verify the appropriate properties of $F(\lambda)$ for $\lambda_{-} = 0$ and $\lambda _{+} = \infty$.

Let $X_{1} := \begin{bmatrix} 1 & 0 \\ 0 & 0 \end{bmatrix}$, $X_{2} := \begin{bmatrix} 0 & 0 \\ 0 & 1 \end{bmatrix}$, $g_{\lambda} := \left[ \left( 1-q_{1}\mathcal{PV}f_{\lambda} \right)^{2} + \left(q_{1} \pi w_{\lambda} \right)^{2} \right]^{-1}$, and
\begin{align}
	\widehat{g}_{1,\lambda} :=&\ \left[ \left( 1-q^{2}_{2}\mathcal{PV}f^{L}_{\lambda_{1}}f^{L}_{\lambda_{2}} \right)^{2} + \left( \pi w^{L}_{\lambda_{1}}q^{2}_{2}f^{L}_{\lambda_{2}} \right)^{2} \right]^{-1}, \\
	\widehat{g}_{2,\lambda} :=&\ \left[ \left( 1-q^{2}_{2}f^{L}_{\lambda_{1}}\mathcal{PV}f^{L}_{\lambda_{2}} \right)^{2} + \left( \pi w^{L}_{\lambda_{2}}q^{2}_{2}f^{L}_{\lambda_{1}} \right)^{2} \right]^{-1} .
\end{align}

Since $\psi^{L}_{z} = (1 - q_{1}f_{z})^{-1}\psi_{z}$, $f^{L}_{z} = (1 - q_{1}f_{z})^{-1}f_{z}$, and $\psi_{z} = f_{z}\phi_{z}+\xi_{z}$, for $\lambda \ge 0$ one has
\begin{align}
	\mathcal{PV}\psi^{L}_{\lambda} &= g_{\lambda}\left[ \mathcal{PV}f_{\lambda}\phi_{\lambda} - q_{1}\mathcal{PV}f_{\lambda}\xi_{\lambda} + \xi_{\lambda} - q_{1}\left( \mathcal{PV}f_{\lambda} \right)^{2}\phi_{\lambda} - q_{1}(\pi w_{\lambda})^{2}\phi_{\lambda} \right], \\
	\mathcal{PV}f^{L}_{\lambda} &= g_{\lambda}\left[ \mathcal{PV}f_{\lambda} - q_{1}\left( \mathcal{PV}f_{\lambda} \right)^{2} - q_{1}(\pi w_{\lambda})^{2} \right], \\
	\phi^{L}_{\lambda} &= \phi_{\lambda} + q_{1}\xi_{\lambda}, \quad w^{L}_{\lambda} = g_{\lambda}w_{\lambda} ,
\end{align}
and for $\lambda < 0$ one has
\begin{align}
	\psi^{L} = (1 - q_{1}f_{\lambda})^{-1}\psi_{\lambda},\quad f^{L} = (1 - q_{1}f_{\lambda})^{-1}f_{\lambda} .
\end{align}

By the method of spectral shifts one has
\begin{align}
	R^{H_{2}}_{z} &= R^{H_{1}}_{z} + R^{H_{1}}_{z}q_{2}P_{0}JR^{H_{2}}_{z} \\
		&= R^{H_{1}}_{z} + R^{H_{1}}_{z}q_{2}P_{0}J\left(1-f^{H_{1}}_{z}q_{2}J\right)^{-1}R^{H_{1}}_{z} \\
		&= R^{H_{1}}_{z} + R^{H_{1}}_{z}q_{2}P_{0}J\left(1-q^{2}_{2}f^{L}_{z_{1}}f^{L}_{z_{2}}\right)^{-1}\left(1+f^{H_{1}}_{z}q_{2}J\right)R^{H_{1}}_{z} \\
		&= R^{H_{1}}_{z} + q_{2} \left| 1-q^{2}_{2}f^{L}_{z_{1}}f^{L}_{z_{2}} \right|^{-2} \left( 1-q^{2}_{2}\overline{f}^{L}_{z_{1}}\overline{f}^{L}_{z_{2}} \right) R^{H_{1}}_{z}\left(J+q_{2}Jf^{H_{1}}_{z}J\right)P_{0}R^{H_{1}}_{z} \\
		&= \left( R^{L}_{z_{1}}X_{1} - q_{2} R^{L}_{z_{2}}X_{2} \right) + \left| 1-q^{2}_{2}f^{L}_{z_{1}}f^{L}_{z_{2}} \right|^{-2} \left( 1-q^{2}_{2}\overline{f}^{L}_{z_{1}}\overline{f}^{L}_{z_{2}} \right) \\
			&\quad\quad \times \left( R^{L}_{z_{1}}X_{1} - R^{L}_{z_{2}}X_{2} \right) \left[J - q_{2}\left( f^{L}_{z_{1}}X_{2} - f^{L}_{z_{2}}X_{1} \right)\right] P_{0} \left( R^{L}_{z_{1}}X_{1} - R^{L}_{z_{2}}X_{2} \right)
\end{align}
from which one finds for $\lambda \ge \mu$:
\begin{align}
	\lim_{\epsilon \searrow 0} R^{H_{2}}_{\lambda \pm i\epsilon} &= \left[ \left( \mathcal{PV}^{L}_{\lambda_{1}}X_{1} - R^{L}_{\lambda_{2}}X_{2} \right) \pm i\pi \delta^{L}_{\lambda_{1}}X_{1} \right] \\
			&\quad\quad + q_{2} \widehat{g}_{1,\lambda} \left[ \left( 1-q^{2}_{2}\mathcal{PV}f^{L}_{\lambda_{1}}f^{L}_{\lambda_{2}} \right) \pm i\pi w^{L}_{\lambda_{1}}q^{2}_{2}f^{L}_{\lambda_{2}} \right] \\
			&\quad\quad \times \left[ \left( \mathcal{PV}^{L}_{z_{1}}X_{1} - R^{L}_{z_{2}}X_{2} \right) \pm i\pi \delta^{L}_{\lambda_{1}}X_{1} \right] \\
			&\quad\quad \times \left[ \left( J - q_{2} \mathcal{PV}f^{L}_{\lambda_{1}}X_{2} + q_{2} f^{L}_{\lambda_{2}}X_{1} \right) \mp i\pi w^{L}_{\lambda_{1}}q_{2}X_{2} \right] \\
			&\quad\quad \times P_{0} \left[ \left( \mathcal{PV}^{L}_{\lambda_{1}}X_{1} - R^{L}_{\lambda_{2}}X_{2} \right) \pm i\pi \delta^{L}_{\lambda_{1}}X_{1} \right]
\end{align}
\begin{align}
		&= \left[ \left( \mathcal{PV}^{L}_{\lambda_{1}}X_{1} - R^{L}_{\lambda_{2}}X_{2} \right) \pm i\pi w^{L}_{\lambda_{1}} \phi^{L}_{\lambda_{1}} \otimes \phi^{L,*}_{\lambda_{1}}  X_{1} \right] \\
			&\quad\quad+ q_{2} \widehat{g}_{1,\lambda} \left[ \left( 1-q^{2}_{2}\mathcal{PV}f^{L}_{\lambda_{1}}f^{L}_{\lambda_{2}} \right) \pm i\pi w^{L}_{\lambda_{1}}q^{2}_{2}f^{L}_{\lambda_{2}} \right] \\
			&\quad\quad \times \left[ \left( \mathcal{PV}\psi^{L}_{\lambda_{1}}X_{1} - \psi^{L}_{\lambda_{2}}X_{2} \right) \pm i\pi w^{L}_{\lambda_{1}}\phi^{L}_{\lambda_{1}}X_{1} \right] \\
			&\quad\quad \left[ \left( J - q_{2} \mathcal{PV}f^{L}_{\lambda_{1}}X_{2} + q_{2} f^{L}_{\lambda_{2}}X_{1} \right) \mp i\pi w^{L}_{\lambda_{1}}q_{2}X_{2} \right] \\
			&\quad\quad \otimes \left[ \left( \mathcal{PV}\psi^{L,*}_{\lambda_{1}}X_{1} - \psi^{L,*}_{\lambda_{2}}X_{2} \right) \pm i\pi w^{L}_{\lambda_{1}}\phi^{L,*}_{\lambda_{1}}X_{1} \right]
\end{align}
and for $\lambda \le \mu$:
\begin{align}
	\lim_{\epsilon \searrow 0} R^{H_{2}}_{\lambda \pm i\epsilon} &= \left[ \left( R^{L}_{\lambda_{1}}X_{1} - \mathcal{PV}^{L}_{\lambda_{2}}X_{2} \right) \mp i\pi \delta^{L}_{\lambda_{2}}X_{2} \right] \\
			&\quad\quad + q_{2} \widehat{g}_{2,\lambda} \left[ \left( 1-q^{2}_{2}f^{L}_{\lambda_{1}}\mathcal{PV}f^{L}_{\lambda_{2}} \right) \mp i\pi w^{L}_{\lambda_{2}}q^{2}_{2}f^{L}_{\lambda_{1}} \right] \\
			&\quad\quad \times \left[ \left( R^{L}_{\lambda_{1}}X_{1} - \mathcal{PV}^{L}_{\lambda_{2}}X_{2} \right) \mp i\pi \delta^{L}_{\lambda_{2}}X_{2} \right] \\
			&\quad\quad \left[ \left( J - q_{2} f^{L}_{\lambda_{1}}X_{2} + q_{2} \mathcal{PV}f^{L}_{\lambda_{2}}X_{1} \right) \pm i\pi w^{L}_{\lambda_{2}}q_{2}X_{1} \right] \\
			&\quad\quad \times P_{0} \left[ \left( R^{L}_{\lambda_{1}}X_{1} - \mathcal{PV}^{L}_{\lambda_{2}}X_{2} \right) \mp i\pi \delta^{L}_{\lambda_{2}}X_{2} \right] \\
		&= \left[ \left( R^{L}_{\lambda_{1}}X_{1} - \mathcal{PV}^{L}_{\lambda_{2}}X_{2} \right) \mp i\pi w^{L}_{\lambda_{2}} \phi^{L}_{\lambda_{2}} \otimes \phi^{L,*}_{\lambda_{2}}X_{2} \right] \\
			&\quad\quad + q_{2} \widehat{g}_{2,\lambda} \left[ \left( 1-q^{2}_{2}f^{L}_{\lambda_{1}}\mathcal{PV}f^{L}_{\lambda_{2}} \right) \mp i\pi w^{L}_{\lambda_{2}}q^{2}_{2}f^{L}_{\lambda_{1}} \right] \\
			&\quad\quad \times \left[ \left( \psi^{L}_{\lambda_{1}}X_{1} - \mathcal{PV}\psi^{L}_{\lambda_{2}}X_{2} \right) \mp i\pi w^{L}_{\lambda_{2}} \phi^{L}_{\lambda_{2}}X_{2} \right] \\
			&\quad\quad \left[ \left( J - q_{2} f^{L}_{\lambda_{1}}X_{2} + q_{2} \mathcal{PV}f^{L}_{\lambda_{2}}X_{1} \right) \pm i\pi w^{L}_{\lambda_{2}}q_{2}X_{1} \right] \\
			&\quad\quad \otimes \left[ \left( \psi^{L}_{\lambda_{1}}X_{1} - \mathcal{PV}\psi^{L}_{\lambda_{2}}X_{2} \right) \mp i\pi w^{L}_{\lambda_{2}} \phi^{L,*}_{\lambda_{2}}X_{2} \right] .
\end{align}
One may expand the above expressions and look for the resulting imaginary piece to find $\delta^{H_{2}}_{\lambda}$ for  $\lambda \ge \mu$ or $\lambda \le \mu$. We will not do so and will analyze its properties from the unexpanded forms for simplicity instead.

We will use the definitions for $W_{\kappa,\tau}$, $\epsilon$, and the like from \cite{paper 01}. Furthermore we will employ the spectral decay estimates of Corollary 1 of \cite{paper 01} as well as the quasi-exponential decay estimates of Theorem 2 of \cite{paper 02}. From the latter one can see that for $\lambda > \mu$ it is the case that $\mathrm{d}^{n}_{\lambda}\left[w^{1/2}_{\lambda_{1}}\psi_{\lambda_{2}}(x) \right] \in \ell^{1}$, as a function of $(x,\lambda)$, is $\ell^{1}(\mathbb{Z}_{+} \times [\mu,\infty),\mathbb{R})$.

By considering the many definitions, there is one crucial function which strongly determines our estimates: $\mathcal{PV}f_{a} = e^{a}E_{1}(a) \sim - \log(a)$ as $0 < a \searrow 0$. Only powers of $\mathcal{PV}f_{a}$ can be nonanalytic or unbounded. We will therefore proceed to prove the desired properties of $F(\lambda)$ by addressing the powers of $\mathcal{PV}f_{a}$ alone.

One may observe that $g_{\lambda} := \{ [1- q_{1} e^{-\lambda} \mathcal{PV}E_1(-\lambda)]^2 + [\pi q_{1} e^{-\lambda}]^2 \}^{-1}$ has the properties:
\begin{align}
	g_{\lambda} &= | g_{\lambda} | \le \widehat{g}_{0}(q_{1}) < \infty, \quad \forall \lambda \in [0, \infty) \\
	| \mathrm{d}_{\lambda}g_{\lambda} | &\le \widehat{g}_{0}(q_{1})\widehat{g}_{1}(q_{1}, \delta) < \infty,  \quad \forall \lambda \in [\delta, \infty) \\
	g_{0} &= g_{\infty} = 0 , \\
	\mathrm{d}_{\lambda}g_{\lambda} &= \mathcal{O}(\lambda^{-1}\log^{-1}\lambda), \quad \lambda \searrow 0 \\
	\mathrm{d}_{\lambda}^{2}g_{\lambda} &= \mathcal{O}(\lambda^{-2}\log^{-3}\lambda), \quad \lambda \searrow 0 \\
		&\leq \mathcal{O}(\lambda^{-2}\log^{-2}\lambda)
\end{align}
where $0 < \widehat{g}_{0}(q_{1}), \widehat{g}_{1}(q_{1}, \delta) < \infty$ are constants whose other properties are not needed here. $g_{\lambda}$ is the only function of $\lambda$ involved in the definition of $F(\lambda)$ whose derivatives are unbounded in the neighborhood of the threshold $\lambda = 0$ and thereby the derivatives of $g_{\lambda}$ and positive powers of $\mathcal{PV}f_{\lambda}$ and its derivatives are dominant in determining the properties of the derivatives of $F(\lambda)$. We will therefore only consider the dominant factors with respect to these quantities.

Consider the contributions to the imaginary part of $\lim_{\epsilon \searrow 0} R^{H_{2}}_{\lambda \pm i\epsilon}$ for either $\lambda \ge \mu$ or $\lambda \le \mu$. Due to the symmetry between these two ranges of $\lambda$ it is sufficient to analyze the case of $\lambda \ge \mu$ alone and we will do so exclusively in the following.

\noindent\underline{Properties (1), (2):}
By considering the control that the factors of $w^{L}_{\lambda_{1}} = g_{\lambda_{1}}w_{\lambda_{1}}$ impose, one may see that all possible contributions are bounded in $\lambda_{1}$ and $x$.

\noindent\underline{Property (3):}
The bounds for the derivatives of $\phi_{\lambda_{1}}$, $\xi_{\lambda_{1}}$, and $\psi_{\lambda_{2}}$ cannot present a problem with the chosen norm on $\mathscr{B}$. Exponential decay as $\lambda_{1} \nearrow \infty$ ensures that the upper bound of integration cannot be a problem. The only remaining potential issue comes from the behavior at the threshold, which is not relevant.

\noindent\underline{Property (4):}
The bounds for the derivatives of $\phi_{\lambda_{1}}$, $\xi_{\lambda_{1}}$, and $\psi_{\lambda_{2}}$ cannot present a problem with the chosen norm on $\mathscr{B}$.There will be two dominant factors. One of is $\mathrm{d}_{\lambda}g_{\lambda_{1}} = \mathcal{O}(\lambda^{-1}_{1}\log^{-1}\lambda_{1})$ as $\lambda_{1} \searrow 0$. The other dominant factor is of the form $g^{2}_{\lambda_{1}}\mathrm{d}_{\lambda_{1}}(\mathcal{PV}f_{\lambda_{1}})^{2} = \mathcal{O}(\lambda^{-1}_{1}\log^{-1}\lambda_{1})$ as $\lambda_{1} \searrow 0$.

\noindent\underline{Property (5):}
The bounds for the derivatives of $\phi_{\lambda_{1}}$, $\xi_{\lambda_{1}}$, and $\psi_{\lambda_{2}}$ cannot present a problem with the chosen norm on $\mathscr{B}$. There will be one dominant factor, which is $g_{\lambda_{1}}[\mathrm{d}_{\lambda_{1}}(\mathcal{PV}f_{\lambda_{1}})^{2}]^{2} = \mathcal{O}(\lambda^{-2}\log^{-2}\lambda)$.

\end{proof}

\begin{proof}[Proof of Theorem \ref{snthm08}]
One has that the linear Schr\"odinger equation
\begin{align}
	i\mathrm{d}_{t}u &= Hu = H_{2}u + Uu
\end{align}
may be converted to an integral equation by means of the Duhamel formula:
\begin{align}
	u(t) &= e^{-itH_{2}}u_{0} - i\int_{0}^{t}\mathrm{d}t_{1}\ e^{-i(t - t_{1})H_{2}}Uu(t),
\end{align}
to which one may apply weights and estimate
\begin{align}
	W_{\kappa,\tau}u(t) &= W_{\kappa,\tau}e^{-itH_{2}}u_{0} - i\int_{0}^{t}\mathrm{d}t_{1}\ W_{\kappa,\tau}e^{-i(t - t_{1})H_{2}}Uu(t_{1}) \\
	|| W_{\kappa,\tau}u(t) ||_{1} &\le || W_{\kappa,\tau}e^{-itH_{2}}u_{0} ||_{1} + \int_{0}^{t}\mathrm{d}t_{1}\ || W_{\kappa,\tau}e^{-i(t - t_{1})H_{2}}Uu(t_{1}) ||_{1} \\
	&\le c_{0}|| W_{\kappa,\tau}e^{-itH_{2}}W_{\kappa,\tau}W^{-1}_{\kappa,\tau}u_{0} ||_{\infty} \\
			&\quad\quad + c_{1}\int_{0}^{t}\mathrm{d}t_{1}\ || W_{\kappa,\tau}e^{-i(t - t_{1})H_{2}}W_{\kappa,\tau} W^{-1}_{\kappa,\tau}UW^{-1}_{\kappa,\tau}W_{\kappa,\tau}u(t_{1}) ||_{\infty} \\
		&\le c_{0}|| W_{\kappa,\tau}e^{-itH_{2}}W_{\kappa,\tau} ||\ || W^{-1}_{\kappa,\tau}u_{0} ||_{1} \\
			&\quad\quad + c_{1}\int_{0}^{t}\mathrm{d}t_{1}\ || W_{\kappa,\tau}e^{-i(t - t_{1})H_{2}}W_{\kappa,\tau}||\ || W^{-2}_{\kappa,\tau}U ||\ || W_{\kappa,\tau}u(t_{1}) ||_{1} \\	
		&\le c_{2}[(t + c_{3})\log^{2}(t + c_{3})]^{-1} \\
			&\quad\quad + c_{4}\int_{0}^{t}\mathrm{d}t_{1}\ [(t - t_{1} + c_{3})\log^{2}(t - t_{1} + c_{3})]^{-1} || W_{\kappa,\tau}u(t_{1}) ||_{1}
\end{align}
Let $f(t) := || W_{\kappa,\tau}u(t) ||_{1}$ and $g(t) := [(t + c_{3})\log^{2}(t + c_{3})]^{-1}$. By Gronwall's Lemma one has
\begin{align}
	f(t) &\le c_{2}g(t) + c_{4}\int_{0}^{t} \mathrm{d}t_{1}\ g(t-t_{1})f(t_{1})
\end{align}
implies that
\begin{align}
	f(t) &\le c_{2}g(t) + c_{2}c_{4}\int_{0}^{t} \mathrm{d}t_{1}\ g(t_{1})g(t - t_{1})\exp\left[c_{4}\int_{t_{1}}^{t} \mathrm{d}t_{2}\ g(t - t_{2})\right] .
\end{align}
Furthermore
\begin{align}
	\exp\left[c_{4}\int_{t_{1}}^{t} \mathrm{d}t_{2}\ g(t - t_{2})\right] &= \exp\left[c_{4}\int_{0}^{t-t_{1}} \mathrm{d}t_{2}\ g(t_{2})\right] \\
	&\le \exp\left[c_{4}\int_{0}^{\infty} \mathrm{d}t_{2}\ g(t_{2})\right]  \le c_{5} ,
\end{align}
and
\begin{align}
	\int_{0}^{t} \mathrm{d}t_{1}\ g(t_{1})g(t - t_{1}) &= 2\int_{0}^{t/2} \mathrm{d}t_{1}\ g(t_{1})g(t - t_{1}) \le 2\int_{0}^{t/2} \mathrm{d}t_{1}\ g(t_{1})g(t/2) \\
		&\le c_{6} [(t + c_{7})\log^{2}(t + c_{7})]^{-1} .
\end{align}
One therefore has
\begin{align}
	f(t) &\le c_{2}g(t) + c_{2}c_{4}c_{5}c_{6} [(t + c_{7})\log^{2}(t + c_{7})]^{-1} \\
		&\le c_{8}[(t + c_{9})\log^{2}(t + c_{9})]^{-1} = \mathcal{O}(t^{-1}\log^{-2}t),
\end{align}
as $t \nearrow \infty$. In the above $c_{j}$, $j = 0, \ldots , 9$, are constants, the properties of which are not important.
\end{proof}

\section{Conclusions and Conjectures}

The goal of this work has been to prove Theorem \ref{snthm08}. This is the most important component in proving the asymptotic stability of the soliton:
\begin{snconj}\label{snconj01}
The soliton manifold specified by the coordinates $(\mu,\nu)$ with respect to $u(t) = e^{-i(-\mu t + \nu)}\alpha_{\mu}$ is asymptotically stable under perturbed evolution via the discrete NLS \eqref{NLS}.
\end{snconj}
The remaining proof of this conjecture may be sketched as follows. The Duhamel formula gives an expression for the evolution of the perturbation $\beta$:
\begin{align}\label{radiation}
	\vec{\beta}(t_1) &= e^{-it_1H}\vec{\beta}_{0} - i\int_{0}^{t_1}\mathrm{d}t_{2}\ e^{-i(t_1 - t_{2})H}\vec{\gamma}(t_2) ,
\end{align}
where $\vec{\beta}_{0} = \vec{\beta}(t = 0)$. This alone is not enough to determine the evolution of $\beta$ as the parameters $\widehat{\mu}$ and $\widehat{\nu}$ are time dependent. One must include separate evolution equations for these as well. First one assumes that $(\vec{\alpha},\vec{\beta}(0))_{\vec{\mathscr{H}}} = (\vec{\alpha},\mathrm{d}_{t}\vec{\beta}(0))_{\vec{\mathscr{H}}} = 0$, where $(\cdot,\cdot)_{\vec{\mathscr{H}}}$ is natural the extension of the inner product of $\mathscr{H}$ to the matrix system and $\vec{\alpha} = \begin{bmatrix} \alpha \\ -\alpha \end{bmatrix}$. This condition ensures that $\vec{\beta}$ remains in the span of the generalized eigenvectors of $H$. Then one takes the inner product of $\vec{\alpha}$ with both sides of the LNLS \eqref{LNLS} to arrive at
\begin{align}\label{modulation}
	0 = (\alpha,\gamma) \quad \Rightarrow \quad \mathrm{d}_{t}\widehat{\nu} = (\alpha,\alpha_{\widehat{\mu}})^{-1}(\alpha,\Re\gamma_{1}),\quad \mathrm{d}_{t}\widehat{\mu} = i(\alpha,\partial_{\widehat{\mu}}\alpha_{\widehat{\mu}})^{-1}(\alpha,\Im\gamma_{1}) .
\end{align}
Equations \eqref{radiation} and \eqref{modulation} together constitute the \emph{modulation equations} for the NLS \eqref{NLS}, where $\widehat{\mu}$ and $\widehat{\nu}$ are the \emph{modulation parameters} \cite{Avy NLS}.

Due to the work of Soffer and Weinstein in \cite{Avy NLS II} It is reasonable to assume that such a claim is true and it should be the case that one can prove it with an application of bootstrapping estimates. It their analysis it was shown that the linearized Hamiltonian strongly determines the evolution of the system and that obtaining its appropriate weighted $\ell^1 \rightarrow \ell^\infty$ estimates is sufficient to prove asymptotic stability for sufficiently small perturbations. The location of eigenvalues is also crucial for the study of the dynamics. The presence of imaginary eigenvalues indicates an exponential instability in time. For each real eigenvalue one must consider a separate modulation parameter and eigenfunction, the dynamics of which must be included in the modulation equations, thereby further complicating the problem. The full analysis of the spectrum of the linearized Hamiltonian is typically prohibitively difficult. The corresponding work on the 3D continuum radial NLS has recently come to a close after extensive collaborative effort. Please see O. Costin, M. Huang, and W. Schlag \cite{3D NLS} for the conclusion of work on that system.

The case of the real Nonlinear Klein-Gordon equation is likely to be much harder to address. Consider the discrete real Nonlinear Klein-Gordon equation (rNLKG) specified by
\begin{align}\label{rNLKG}
	-\partial^{2}_tu = L_0u - u^{p},\quad 1 < p \in \mathbb{Z}.
\end{align}
It was this equation that was first studied in the context of noncommutative field theory, see e.g. \cite{GMS}, and the mathematical analysis of \cite{CFW}, \cite{DJN 1}, and \cite{DJN 2}. There is an approach to the rNLKG which is similar to that of the methods of linearization taken with the method of modulation equations, but it is of a different character. There, the analogue of the stationary solution will be of the form $u(t) = \cos(\mu t + \nu)\alpha_{\mu}$. Due to the presence of nonlinearity this will not be a stationary solution in general. For the NLS one could interpret the stationary state as a nonlinear variant of the evolution of an eigenfunction with an associated isolated eigenvalue. For the rNLKG one is typically lead to interpret the ``quasi-stationary state'' as the nonlinear variant of the evolution of a resonance function with an associated embedded eigenvalue. The coupling of the ``radiation'' $\beta$ to the soliton will introduce an instability and lead to a resonance with a decay time. This picture was introduced and elaborated upon in the work of Soffer and Weinstein on nonlinear resonances and the nonlinear Fermi-Golden Rule \cite{Avy NLKG}.

Chen, Fr\"ohlich, and Walcher in \cite{CFW} conjecture that Equation \eqref{rNLKG} has localized metastable solutions. This leads us to the following conjecture.

\begin{snconj}\label{snconj02}
Solutions of the discrete rNLKG \eqref{rNLKG} which begin close to $u(t) = \cos(\mu t + \nu)\alpha_{\mu}$ are metastable resonance functions.
\end{snconj}

We seek to address compare and contrast the proofs of these conjectures in future work.

\appendix

\section{Kato's notion of norm resolvent convergence}

We will review the arguments for the proof of Kato's notion of norm resolvent convergence as is presented in \cite{Kato}. The components below refer to yet other components of \cite{Kato} not included here. The full span of arguments needed to fill in all of the requirements of the proof would be beyond the scope of this appendix. We have included what may be considered the most immediately necessary pieces. We note that all operators which we use in the body of this work are closable and therefore may be replaced with their closures, if necessary, without loss of generality.

\begin{unthm}[IV-2.25 of \cite{Kato}, p.\ 206]
Let $A \in \mathcal{C}(\mathscr{B})$ be a closed operator on a Banach space, $\mathscr{B}$, and have a non-empty resolvent set $\rho(A)$. In order that a sequence $A_{n}$ of closed operators converge to $A$ in the generalized sense, it is necessary that each $z \in \rho(A)$ belong to $\rho(A_{n})$ for sufficiently large $n$ and
\begin{align} \label{2.33}
	||R^{A_{n}}_{z} - R^{A}_{z}|| \to 0 ,
\end{align}
while it is sufficient that this be true for some $z \in \rho(A)$.
\end{unthm}


\noindent ``Converge in the generalized sense'' means that the graphs of the two operators converge in the ``gap distance norm'', which is approximately the maximum geometric distance between the graphs of the operators as submanifolds of the extended Banach space.

\begin{rem}[IV-3.13 of \cite{Kato}, p.\ 211]
Theorem IV-3.12 shows explicitly that if $||R^{A}_{z} - R^{B}_{z}||$ is small for some $z$, then it is small for every $z$. More precisely, for any $A \in \mathcal{B}(\mathcal{\mathscr{B}})$ a bounded linear operator on a Banach space $\mathscr{B}$ and $z,z_{0} \in \rho(A)$, then there is a constant $c$ such that
\begin{align}\label{3.9}
	||R^{A}_{z} - R^{B}_{z}|| \le c ||R^{A}_{z_{0}} - R^{B}_{z_{0}}||
\end{align}
for any $B \in \mathcal{B}(\mathscr{B})$ for which $z_{0} \in \rho(B)$ and $ ||R^{A}_{z_{0}} - R^{B}_{z_{0}}||$ is sufficiently small (then $z \in \rho(B)$ is a consequence). This is  another proof of a remark given after Theorem IV-2.25.
\end{rem}

\textbf{Problem} (IV-3.14 of \cite{Kato}, p.\ 212). A more explicit formula than \eqref{3.9} is
\begin{align}
	||R^{A}_{z} - R^{B}_{z}|| &\le \left[ 1 - |z-z_{0}| || (A-z)^{-1}(A-z_{0}) ||\ ||R^{A}_{z_{0}} - R^{B}_{z_{0}}|| \right]^{-1} \\
		&\quad\quad \times || (A-z)^{-1}(A-z_{0}) ||^{2}\ ||R^{A}_{z_{0}} - R^{B}_{z_{0}}|| ,
\end{align}
which is valid if $||^{2}\ ||R^{A}_{z_{0}} - R^{B}_{z_{0}}||$ is so small that the denominator on the right is positive. Here $(A-z)^{-1}(A-z_{0})$ is a convenient expression for $(A-z)^{-1}(A-z_{0}) = (A - z_{0})R^{A}_{z} = 1 + (z - z_{0})R^{A}_{z}$.

By considering the above arguments one may then arrive at the statement of the propositon.

\begin{unprop}[Kato's notion of norm resolvent convergence \cite{Kato}, p. 427]\label{nrc}
Let $\{A_{n}\}_{n=0}^{\infty}$ be a sequence of closed operators on a Banach space $\mathscr{B}$. If $(A_{n} - z)^{-1}$ converges in norm to $(A - z)^{-1}$ as $n \nearrow \infty$ for $A$ closed and for some $z \in \rho(A)$, then the same is true for every $z \in \rho(A)$.
\end{unprop}

\thanks{We thank Marius Beceanu for his helpful discussions. This work was supported in part by the NSF grant DMS 1201394.}

\end{document}